\documentclass{winnower}
\usepackage[T1]{fontenc}
\usepackage[utf8]{inputenc}
\usepackage{dsfont}
\usepackage{pifont}
\usepackage{amsthm}
\usepackage{nicefrac}
\usepackage{xcolor} 
\usepackage{amsmath}
\usepackage{natbib}
\usepackage{graphicx}
\newtheorem{theorem}{Theorem}
\newtheorem{lemma}{Lemma}
\newtheorem{example}{Example}
\newtheorem{assumption}{Assumption}
\newtheorem{definition}{Definition}

\newtheorem{observation}{Observation}
\newtheorem{remark}{Remark}

\newcommand{\ind}{\perp\!\!\!\!\perp} 

\begin{document}

\title{Tell Me Why: Incentivizing Explanations}

\author{Siddarth Srinivasan}
\affil{Harvard University, Cambridge, MA, USA}
\author{Ezra Karger}
\affil{Federal Reserve Bank of Chicago, IL, USA}
\author{Michiel Bakker}
\affil{Massachusetts Institute of Technology, Cambridge, MA, USA}
\author{Yiling Chen}
\affil{Harvard University, Cambridge, MA, USA}

\date{}

\maketitle

\begin{abstract}
Common sense suggests that when individuals explain why they believe something, we can arrive at more accurate conclusions than when they simply state what they believe.  Yet, there is no known mechanism that provides incentives to elicit explanations for beliefs from agents. This likely stems from the fact that standard Bayesian models make assumptions (like conditional independence of signals) that preempt the need for explanations, in order to show efficient information aggregation.  A natural justification for the value of explanations is that agents' beliefs tend to be drawn from overlapping sources of information, so agents' belief reports do not reveal all that needs to be known. Indeed, this work argues that rationales—explanations of an agent’s private information—lead to more efficient aggregation by allowing agents to efficiently identify what information they share and what information is new. Building on this model of rationales, we present a novel `deliberation mechanism' to elicit rationales from agents in which truthful reporting of beliefs and rationales is a perfect Bayesian equilibrium. 

\end{abstract}

\section{Introduction}
From political journalism to major business forecasts, experts routinely offer more than just numerical estimates: they explain how they arrived at those estimates. Outside formal `expert' settings, daily deliberation—whether corporate board discussions or online communities—also involves participants sharing \emph{why} they believe something, not just \emph{what} they believe. Common sense suggests that understanding these underlying reasons can yield more accurate conclusions than possible by aggregating raw beliefs alone.

While mechanisms like proper scoring rules \citep{gneiting2007strictly} and prediction markets \citep{wolfers2004prediction} can provide incentives to elicit beliefs from experts, there are no known mechanisms that can provide incentives to elicit \emph{natural-language rationales} despite the intuitive benefits. This likely stems from the fact that capturing the value of explanations or `rationales' with standard Bayesian models is challenging. Indeed, standard models that demonstrate efficient information aggregation solely from eliciting beliefs assume that beliefs convey everything an agent has to convey. When models do permit agents' beliefs to not fully reveal private information or have informational overlaps, we typically lose efficient information aggregation. The puzzle remains: how can we formally model the commonsense idea that explaining one’s rationale improves aggregation?

One natural explanation for the intuition that rationales or explanations can produce more accurate collective beliefs is this: agents' beliefs are drawn from overlapping sources of information, and sharing rationales helps agents `discount' overlapping or shared pieces of information to prevent `double-counting.' The idea of agents' drawing beliefs from partially overlapping sources has been considered in social learning \citep{dasaratha2019aggregative}, forecast aggregation \citep{satopaa2017partial, babichenko2021learning}, and epistemic social choice \citep{ding2021deliberation, dietrich2024deliberation}. While these works study the dynamics of information aggregation under such models, we are interested in applying such a model towards a mechanism to elicit explanations for agents' beliefs.

To this end, we propose a conceptually simple, tractable model where rationales allow agents to identify the aspects of their beliefs that are derived from shared information and thus enable efficient aggregation. Specifically, we model agents' beliefs as composed of `atomic' pieces that may overlap across agents. Then, we consider a setting where agents arrive sequentially and observe all previous agents' reports. If previous agents reported only their belief, they only reveal a lossy aggregate of these overlapping pieces. However, if previous agents reported rationales, they also revealed the specific pieces composing their belief, allowing future agents to efficiently disregard already-known information. In this way, rationales permit more efficient Bayesian aggregation. We show that providing rationales increases the speed and limit of information aggregation relative to belief-only reports.

Building on this model, we develop a `deliberation mechanism' to incentivize agents to submit both their beliefs and the rationales that explain these beliefs. In most settings, composing a textual or structured rationale is costly: it requires time and effort and thus must be incentivized. In our mechanism, `experts' with private information sequentially report their belief and rationale (after observing all previous reports) to a `supervisor' who then reports their own belief at each time-step to the principal.  We use a log scoring rule to reward each expert according to how much their report improves the \emph{supervisor's} predictions, and to reward the supervisor for the accuracy of their report at each time-step. Supervisors' beliefs will be more accurate when agents report rationales, but agents have no natural incentive \emph{to} report rationales; after all, simply reporting beliefs induces the same supervisor belief update at no extra effort cost. To overcome this, the supervisor must \emph{commit} ex-ante to ignoring reports without rationales. We show that when the supervisor makes such a commitment, it is a perfect Bayesian equilibrium for agents to truthfully report their belief and rationale.

Thus, the primary contribution of this paper is an incentive-compatible mechanism for eliciting natural-language rationales from agents that explain their beliefs, where truthfully reporting beliefs and rationales is a perfect Bayesian equilibrium. Our mechanism is built on a model of rationales that may also be of independent interest, as it shows how rationales enable efficient information aggregation under the natural setting where agents' beliefs are drawn from overlapping sources of information. Our proposed mechanism has broad relevance to applications from judgmental forecasting \citep{tetlock2016superforecasting}, to forecasting with LLMs \citep{halawi2024approaching, schoenegger2024wisdom, karger2024forecastbench} to AI alignment \citep{hubinger2020ai, bowman2022measuring}.

\paragraph{Outline}  This paper is structured as follows: in Section 2, we review related work; in Section 3, we develop our model of rationales; in Section 4, we show how information aggregates more efficiently with rationales than with mere beliefs; in Section 5, we present our deliberation mechanism to elicit rationales from agents; in Section 6, we discuss other constructions and extensions of our proposed mechanism; and in Section 7, we conclude with directions for future work. Longer proofs are presented in the appendix.

\section{Related Work}

This work is relevant to several different literatures across economics and computer science. We provide an overview of the most relevant works here:

\paragraph{Forecast Aggregation with Overlapping Information} Our work compares the value of rationales against the relatively strong baseline of aggregating just agents' beliefs, a problem extensively studied in the forecast aggregation literature. The most common approach to aggregating raw forecasts is \emph{linear opinion pooling}, which involves taking a weighted average of forecasts \citep{armstrong2001combining}. \citet{ranjan2010combining} and \citet{satopaa2014combining} observe that linear opinion pooling is not calibrated, lacks resolution, and is underconfident when agents' information comes from diverse sources, and hence needs to be \emph{extremized} \citep{satopaa2015combining}.  \citet{Neyman2021AreYS} investigate the properties of averaging forecasts when expert signals are informational substitutes, and show how much forecasts should be extremized when the mechanism has access to a prior.
The need to extremize signals often stems from the fact that agents' forecasts are generated from overlapping sources of information. \citet{palley2019extracting} elicit both private beliefs and expectation of crowd beliefs and use a `pivoting' algorithm that infers what is shared and what is novel to do aggregation. \citet{satopaa2016modeling,satopaa2017partial} and \citet{babichenko2021learning} also study optimal forecast aggregation in settings with partial informational overlap. While we construct similar models, these works focus on the best possible aggregation given agents' signals; we use informational overlaps to motivate the need for rationales to uncover such overlaps.

\paragraph{Epistemic Social Choice and Deliberation} The epistemic social choice literature studies \emph{voting} mechanisms that can successfully aggregate information about the true state of the world.  
Our model of rationales is conceptually similar to deliberation models in epistemic social choice \citep{dietrich2024deliberation, ding2021deliberation}\@: we conceive of rationales as a deliberation mechanism that reveals agents' informational overlaps. As in our work, \citet{dietrich2024deliberation} suggest overlapping Gaussian evidence as a possible instantiation of their model and consider the relative value of increasing group size versus deliberation. Our setting primarily differs in the following ways: (1) we evaluate rationales against a baseline where Bayesian agents communicate their beliefs/signals instead of binary votes. This reveals private information more finely and mitigates phenomena like herding \citep{banerjee1992simple, bikhchandani1998learning}, and so \emph{forecast} aggregation is a more efficient baseline; and (2) we consider a mechanism design problem where agents' do not intrinsically have utilities/preferences over the world states, but instead must be incentivized to share information. There is also significant empirical work on the benefits and pitfalls of deliberation for group decision-making \citep{navajas2018aggregated, graeber2024explanations, lorenz2011social, moshman1998collaborative}. The Delphi method \citep{dalkey1963experimental, helmer1967analysis} is another deliberation method designed to elicit a consensus opinion from a group by alternating between eliciting opinions with justifications and sharing these justifications with the group. Although the method is not theoretically motivated, the structure is motivated by the idea that experts exchanging information can produce more informed opinions. 

\paragraph{Social Learning}  Information aggregation is also studied in the \emph{social learning} \citep{golub2017learning} literature. This literature typically explores models where agents share signals/beliefs and may update their own belief through Bayes' rule \citep{acemoglu2011bayesian} or heuristic approaches like the DeGroot method \citep{degroot1974reaching}. \citet{hkazla2021bayesian} show that Bayesian reasoning over opinions shared in a network is generally difficult. \citet{dasaratha2019aggregative} provide a tractable Gaussian model of social learning that they use to analyze how efficiently a network aggregates information in the presence of informational confounds; we adopt similar techniques to develop a tractable model applicable to our setting.

\paragraph{Aumann's Agreement Theorem} \citet{aumann1976agreeing} and \citet{geanakoplos1982we} present models where Bayesian agents aggregate information merely by sharing beliefs. \citet{kong2022false} give positive results showing that merely exchanging beliefs allows for full information aggregation if signals are independent conditional on ground truth. However, this is a significant assumption that doesn't hold in general since it neglects agents' informational overlap, which limits efficient aggregation from simply exchanging beliefs. Our proposed model shows that aggregation is more efficient when agents reveal explanations, relative to when they simply exchange beliefs as in the Aumannian protocol.

\paragraph{Market-based Information Aggregation} The economics literature has studied how market mechanisms can efficiently aggregate market participants' information via prices \citep{fama1970efficient}. Prediction markets \citep{wolfers2004prediction, chen2010gaming} are a specific market mechanism where agents make predictions about (typically) binary outcomes and are paid out based on the outcome. \citet{ostrovsky2009information} and \citet{kong2022false} show conditions under which prediction markets can aggregate participants' information; in our work, we consider information structures under which prediction markets would not efficiently aggregate information.  \citet{hanson2003combinatorial} proposed \emph{market scoring rules} for prediction markets, that allow us to cast prediction markets as a sequential mechanism where agents are paid using proper scoring rules for the \emph{information added}. We use such ideas to design the payoffs for agents producing rationales.

\paragraph{Elicitation Beyond Forecasts} \citet{frongillo2015elicitation} motivate the need for eliciting \emph{more} than agents' signals/forecasts with a parametric model where the principal seeks to elicit a notion of \emph{confidence} from agents in order to efficiently aggregate information. This is different from our motivation, which is to discount shared information.  \citet{srinivasan2021auctions} propose a mechanism to elicit truthful \emph{peer reviews} in a peer prediction context, where agents report scores and predict other agents' scores on various criteria; in our work, we seek to directly incentivize revealing a rationale.

\paragraph{AI Alignment}
While our model is designed with Bayesian agents in mind, our mechanism to elicit truthful rationales is relevant to ideas in AI alignment, particularly debate methods \citep{khan2024debating, michael2023debate, irving2018ai}. The idea is that incentives for explanations, adversarial back-and-forths, and collaborative discussion could potentially be used as rewards/loss signals in an AI alignment context. Our mechanism is perhaps most similar in spirit to the proposal `AI safety via market making' \citep{hubinger2020ai} where an agent is tasked with maximally changing a principal's belief, although our mechanism is not adversarial. In Section \ref{sec:alt}, we also discuss a `self-resolving' extension of our mechanism applicable in contexts without ground truth that could be applicable to the problem of scalable oversight \citep{bowman2022measuring}.

%-------------------------------------------------%
\section{Model} \label{sec:model}
%-------------------------------------------------%

We consider a \emph{principal} $P$, a \emph{supervisor agent} $S$, and \emph{expert agents} $1, 2, \ldots, T, \ldots$ (generally indexed with $t, t'$).\footnote{We separate the role of the supervisor and principal for clarity of presentation, but these roles may be combined.} The experts have private information that the principal seeks to elicit and aggregate to arrive at an informed belief about binary outcome $Y$. To this end, the principal tasks the supervisor with collecting and aggregating the private information held by the experts and reporting this to the principal. For convenience, we assume the supervisor does not have any private information, though this can easily be relaxed. Experts will arrive sequentially one at a time, and can view the full history of previous reports made to the supervisor. Experts (strategically) report their own beliefs, and importantly, may also (strategically) report a rationale explaining their belief to the supervisor. 

Now, we present our model of experts' information, with the goal of explaining how rationales, when submitted, affect what experts know. Specifically, our model is designed to capture the intuitive idea that an expert's rationale doesn't affect their own belief about outcome $Y$ (which already summarizes what they know); instead, a rationale's utility lies in how it helps \emph{other} experts reason about what information is shared and thus efficiently aggregate information.

\subsection{Setup} 
Let $\Omega$ be the state space, $\mathcal{I} = \{1, 2, \ldots, T, \ldots\}$ represent index of the experts, and $Y$ be the binary outcome of interest. We model experts' private signals and shared signals as constructed from some `atomic' information. For every finite, non-empty subset of experts $S \subset \mathcal{I}$, we define random variable $Z_S: \Omega \rightarrow \mathbb{R}$. We interpret $Z_S$ as a component `sub-signal' that constitutes the information that every expert $t \in S$ possesses and no other expert has. ${\bf Z}_{S} = (Z_{S_1}, Z_{S_2}, \ldots )$, where $S_i \supseteq S$, is the vector of `atomic' pieces of information that construct $Z_S$, the unique information shared only by experts in $S$. $Z_{S_i}$'s are conditionally independent given $Y$. We model ${\bf Z}_{S}$ as a summable sequence, that is, ${\bf Z}_{S}\in l_1$ for all $S\subseteq \mathcal{I}$. Now we can define expert $t$'s \emph{private signal} as $X_t = \sum_{S: t\in S} Z_S = \sum_{i} \left({\bf Z}_{\{t\}}\right)_i$; intuitively an expert $t$'s signal is composed of all the `atoms' $Z_S$ where $t \in S$. We can also define the \emph{shared signal} $X_S = \sum_{S': S \subseteq S'} Z_{S'} = \sum_{i} \left({\bf Z}_{S}\right)_i$ as a component signal that feeds into the private signal of every agent in $S$. Naturally, ${\bf Z}_S$ is a component of ${\bf Z}_t$ for any $t \in S$. We provide a concrete example below:

\begin{example}
    Suppose there are 5 experts, $\mathcal{I} = \{1, 2, 3, 4, 5\}$. $Z_{\{1, 2, 4\}}$ is a random variable that constitutes the signals unique to experts $1, 2,$ and $4$. Then, ${\bf Z}_{\{1, 2, 4\}} = \begin{pmatrix} Z_{\{1,2,4\}}, & Z_{\{1,2,3,4\}}, & Z_{\{1,2,4,5\}}, & Z_{\{1,2,3,4,5\}} \end{pmatrix}$, a vector of all atoms $Z_S$ where $\{1, 2, 4\} \subseteq S$, i.e., all the random variables that make up the signals uniquely seen by experts 1, 2 and 4 only.  We can think of each element of the vector ${\bf Z}_S$ as conditionally independent (given $Y$) atoms used to construct any agent's signal. $X_{\{1,2,4\}} = \sum_i ({\bf Z}_{\{1,2,4\}})_i=Z_{\{1,2,4\}}+Z_{\{1,2,3,4\}}+Z_{\{1,2,4,5\}}+Z_{\{1,2,3,4,5\}}$ is the exact signal that experts 1,2, and 4 share in common, though other experts can see parts of this signal as specified by the atoms in the construction (unlike $Z_{\{1,2,4\}}$ which is information unique to experts 1,2, and 4). Clearly, the atoms $Z_S$ that go into the construction of $X_{\{1,2,4\}}$ are exactly the shared atoms in the construction $X_1$, $X_2$, and $X_{4}$; this is what introduces `informational overlaps' between agents 1, 2, and 4. In a parallel manner, we can construct an agent's observed signal, say $X_4$ as the sum of atoms $Z_S$ where $\{4\} \in S$, i.e., $X_4 = \sum_i ({\bf Z}_{\{4\}})_i$.
\end{example}

%\begin{example}
%    As an example, $Z_{\{1, 7, 34\}}$ is a random variable that constitutes the signals unique to experts $1, 7,$ and $34$. Then, ${\bf Z}_{\{1, 7, 34\}} = \begin{pmatrix} Z_{\{1,7,34,2\}} & Z_{\{1,7,34,3\}} & \ldots & Z_{\{1,7,34,2,3\}} & \ldots & Z_{\{1,7,34,3,5,8\}} & \ldots \end{pmatrix}$, a vector of all atoms $Z_S$ where $\{1, 7, 34\} \subseteq S$, i.e., all the random variables that make up the signals seen by expert 1, expert 7, and expert 34.  We can think of each element of the vector ${\bf Z}_S$ as conditionally independent (given $Y$) atoms used to construct any agent's signal. We write $X_{\{1,7,34\}} = \sum_i ({\bf Z}_{\{1,7,34\}})_i$ as the exact signal agents 1,7,34 share in common, though other agents do see parts of this signal as specified by the atoms in the construction (unlike $Z_{\{1,7,34\}}$ which is information unique to experts 1,7,34). Clearly, the atoms $Z_S$ that go into the construction of $X_{\{1,7,34\}}$ are exactly the shared atoms in the construction $X_1$, $X_7$, and $X_{34}$; this is what introduces `informational overlaps' between agents 1, 7, and 34. In a parallel manner, we can construct an agent's observed signal, say $X_7$ as the sum of atoms $Z_S$ where $7 \in S$, i.e., $X_7 = \sum_i ({\bf Z}_{\{7\}})_i$.
%\end{example}

Next, we denote the \emph{rationale} associated with signal $X_S$ as $\theta_S \in \mathcal{X}$, where $\mathcal{X}$ is an abstract space of rationales. We model the rationale through it's action: as an abstract object that reveals the finer atomic structure of signals $X_S$ when provided. Since experts will submit reports sequentially that can be observed by all future experts, the provision of rationales in these reports modifies the filtration describing how experts' information evolves, by allowing experts to identify what information is shared and what information isn't. Expert $t$'s $\sigma-$algebra given only their private information is $\sigma(X_t)$, but at time-step $t$, they also observe all previous experts' reports, which may include rationales that modify the filtration. Thus, we give the following characterization of rationales:

\begin{definition}[Rationales]
Let $I_t \subseteq \{1, \ldots, t\}$ specify a subset of experts and $I_t^c = \{1, \ldots, t\}\backslash I_t$.  Formally, we define rationales $\{ \theta_{t'} | t' \in I_t\}$ as a set in $\mathcal{X}^{|I_t|}$ that induces the following $\sigma-$algebra at every time-step:
\begin{equation}
    \mathcal{F}_t = \sigma\left(X_S | S \subseteq I_t  \text{ or } S \in \{\{t'\} : t'\in I_t^c\} \right)
\end{equation}
\end{definition}

In other words,  expert $t$'s $\sigma-$algebra is generated by all component shared signals revealed through rationales from previous experts in $I_t$ (note $t \in I_t$ since expert $t$ always knows their own rationale $\theta_t$), as well as from just the raw signals in $I_t^c$. Observe that $\mathbb{F} = \{ \mathcal{F}_t \}_{t \in \mathcal{I}}$ is a filtration specifying the information known to expert $t$ at any time-step $t$. For clarity of presentation, we stick to the filtration $\mathbb{F}^{(1)} = \{\mathcal{F}_t^{(1)}\}$ where all experts provide rationales, and filtration $\mathbb{F}^{(0)} = \{\mathcal{F}_t^{(0)}\}$ where no experts provide rationales. In these cases, the $\sigma-$algebra at time-step $t$ is:
\begin{equation}
    \mathcal{F}_t^{(0)} = \sigma\left(X_t | t \in \{1,\ldots, t\}\right)  \quad \quad \mathcal{F}_t^{(1)} = \sigma\left(X_S | S \subseteq \{1,\ldots, t\}\right) 
\end{equation}

Intuitively, in the filtration with no rationales, experts are only aware of the reported signals; while in the filtration with rationales, experts are able to observe the shared components that constitute their signals. This captures the idea that rationales enable experts to decompose their own information based on how various pieces are shared with other experts. This observation naturally leads to the observation that rationales can lead to more efficient aggregation of information.

\begin{observation}[Rationales can reveal more information than raw signals]
    The filtration $\mathbb{F}^{(1)}$ is a refinement of filtration $\mathbb{F}^{(0)}$ since $\mathcal{F}_t^{(0)} \subseteq \mathcal{F}_t^{(1)}$ for all $t \geq 1$.
\end{observation}

Now, we provide a construction of `new information revealed by a rationale' that will make it easier to analyze the information in filtration $\mathbb{F}^{(1)}$ where rationales are revealed. Specifically, we seek to understand: what information does expert $t$'s rationale reveal that no previous expert revealed? To analyze this, we decompose $X_t = V_t + W_t$ where we refer to $W_t = \sum_{S: t\in S, \{1, \ldots, t-1\} \cap S = \emptyset} Z_S$ as expert $t$'s \emph{residual signal} and $V_t$ as expert $t$'s \emph{redundant signal}. $W_t$ can be interpreted as the component of expert $t$'s signal $X_t$ that is not shared with any of the \emph{previous} experts (but potentially shared with \emph{future} experts), while $V_t$ soaks up all the atoms (and hence dependence) shared with previous experts' signals. In other words, $W_t$ represents the `new information' in rationale $\theta_t$, while $V_t$ represents `old information' already revealed by \emph{some} previous expert. The motivation behind this decomposition is to condense the information revealed by rationales in $\mathcal{F}_t^{(1)}$ into independent chunks $W_1, \ldots, W_t$ that can easily be aggregated to compute posterior beliefs over $Y$. We can also define $W_t$ with an inclusion-exclusion style formula:
\begin{equation}
    W_t = \sum_{\substack{S \subseteq \{1, \ldots, t\} \\ t\in S}} (-1)^{|S|+1} X_S
\end{equation}

\begin{observation}[Properties of residual signals]\label{obs:resprop}
    We note the following properties of residual signals $W_t$:

\begin{enumerate}
\item \textbf{Rationales expose experts' residual signals}: The process $(W_t)_{t\geq 1}$ is adapted to the filtration $\mathbb{F}^{(1)}$, but not adapted to the filtration $\mathbb{F}^{(0)}$.
\item \textbf{The first expert's residual signal is their full signal (by construction)}: $W_1 = X_1$, $V_1 = 0$.
    \item \textbf{Expert $t$'s residual signal is conditionally independent (given $Y$) of previous experts' signals}: $X_{t'} \ind W_t|Y$ and $W_{t'} \ind W_t |Y$  for $1 \leq t' < t$. This is because any $X_{t'}, W_{t'}$ are composed with atoms $Z_{S'}$ where $t' \in S'$, while $t' \notin S$ for atoms $Z_S$ that compose $W_t$. \footnote{Note that $X_{t'} \not\ind W_t|Y$ for $t' > t$, i.e., expert $t$'s `new information' will generally still overlap with future experts' information.}
    \item \textbf{Residual signals $W_1, \ldots, W_t$ are sufficient to compute posterior over $Y$: $\mathbb{P}\left(Y=1|\mathcal{F}_t^{(1)}\right) = \mathbb{P}(Y=1| W_1, \ldots, W_t)$}.  This follows from a simple inductive argument. Since $W_1 = X_1$, it holds in the base case. If it holds at time-step $t-1$, then at time-step $t$,
    {\small
    \begin{equation}
        \mathbb{P}\left(Y=1 | \mathcal{F}_t^{(1)}\right) = \mathbb{P}\left(Y=1 | \mathcal{F}_{t-1}^{(1)} \cup \sigma(X_S | S \subseteq \{1, \ldots, t\}, t\in S \right) = \mathbb{P}(Y=1|W_1, \ldots, W_{t-1}, W_t)
    \end{equation}
    }
    since $Y \ind X_S | W_{1}, \ldots, W_{t-1}$ for $t, t' \in S \subseteq \{1, \ldots, t\}$ and any $t' < t$.
    \end{enumerate}
\end{observation}

Finally, we emphasize two desirable properties of our model: (1) an expert's rationale provides no private benefit, i.e., an expert's own rationale does not refine their own information, even if they observe previous experts' signals (though it may help refine \emph{future} experts' information); (2) rationales are only useful when used in conjunction with \emph{other rationales}; a rationale on its own cannot help identify what information is shared and what information is novel.

\subsection{Information Structure}
The setup above provided an abstract construction of rationales that operationalize the commonsense idea that they reveal more information than raw signals. Now, we impose further structure on experts' information to develop a tractable model to analyze \emph{how much} more information rationales reveal. We begin with standard assumptions on the prior and distribution of experts' signals.

\begin{assumption}[Common prior $\pi$]
    We assume all agents share a common prior over the outcome $\pi = P(Y=1)$ and this prior is common knowledge. We write the prior log odds as $\lambda_\pi = \log\left(\frac{\pi}{1-\pi}\right)$.
\end{assumption}

\begin{assumption}[Conditionally Gaussian signals]
    We assume expert $t$'s signal $X_t$ conditional on outcome $Y$ is distributed as $X_t|(Y=1) \sim \mathcal{N}(\mu, \sigma^2)$ and $X_t|(Y=0) \sim \mathcal{N}(-\mu, \sigma^2)$ with $\mu > 0$.
\end{assumption}

Next, we impose structure on the information shared by any two agents. Specifically, we assume that the shared signal between any two experts $X_{\{t, t'\}}$ is also conditionally normally distributed.

\begin{assumption}[Conditionally Gaussian pairwise shared signals]
    For any pair of experts $\{t, t'\}$, we assume the shared signal $X_{\{t, t'\}}$ conditional on outcome $Y$ is distributed as $X_{\{t, t'\}}|(Y=1) \sim \mathcal{N}(\rho\mu, \rho\sigma^2)$ and $X_{\{t, t'\}}|(Y=0) \sim \mathcal{N}(-\rho\mu, \rho\sigma^2)$ with $\mu > 0$ for some $0 < \rho \leq 1$.
\end{assumption}

As a general notational point, we will write conditional means with the condensed notation $\pm \mu$, where $Y=1$ gives $+\mu$ and $Y=0$ gives $-\mu$. Now, observe that by our construction:
{\small
\begin{equation}
    \text{Cov}(X_t, X_t' | Y) = \text{Cov}\left(\sum_{S:t \in S} Z_S, \sum_{S':t' \in S'} Z_{S'} \left.\right| Y\right) = \sum_{S: t \in S}\sum_{S': t' \in S'} \text{Cov}\left(Z_S, Z_{S'} | Y \right)
\end{equation}
}
From the conditional independence of atomic signals, this reduces $\text{Cov}(X_t, X_t' | Y) = \sum_{\{S: \{t, t'\} \subseteq S\}} \text{Var}\left(Z_S | Y\right) = \text{Var}\left(X_{\{t, t'\}} | Y\right) = \rho\sigma^2$. This then gives:

\begin{observation}[Expert signals as a conditional Gaussian process]
    Given outcome $Y$, the collection of expert signals $\left( X_t \right)_{t\geq 1}$ is a Gaussian process with mean $\pm \mu$ and covariance $\text{Cov}(X_t, X_{t'}) = (\rho + \delta_{t, t'}(1-\rho)) \sigma^2$. For any collection of experts $S \subset \mathcal{I}$, the covariance matrix is $\Sigma_S = \sigma^2((1-\rho)\mathbb{I} + \rho\mathds{1}\mathds{1}^T)$.
\end{observation}

Thus, we have now fully specified the joint distribution of the process $\left( X_t \right)_{t\geq 1}$ which is adapted to the filtration $\mathbb{F}^{(0)}$, and we turn our attention to the process $\left( W_t\right)_{t\geq 1}$ adapted to filtration $\mathbb{F}^{(1)}$. Previous assumptions do not fully determine the distribution of $W_t$ since this depends on the structure of $n$-way overlaps in information; the fact that $(X_1, \ldots, X_t)$ is jointly normal with constant covariance imposes some constraints, but there is still some freedom to specify this overlap. One extreme possibility is that all pairwise overlaps in information are \emph{commonly shared} (i.e., $V_t = V_{t'}$ for all $t, t'$)\footnote{Strictly speaking, our model does not allow $V_t = V_{t'}$ for all $t, t'$ since this would require $V_t = Z_S$ where $S = \mathcal{I}$, but our construction is only defined for finite $S$. Instead, we think of this as the `limiting case.'}, which forces $W_t|Y \sim \mathcal{N}\left((1-\rho)\mu, (1-\rho)\sigma^2\right)$; importantly, note that $\text{Var}(W_t)$ is constant in $t$, so every expert contributes a constant amount of new information indefinitely. The other extreme possibility is that all pairwise overlaps are distinct, so $\text{Var}(W_t | Y) = \max\{ (1-(t-1)\rho)\sigma^2, 0\}$, but this doesn't let us cleanly consider more than $T >  \frac{1}{\rho}$ experts. In order to analyze the case where the amount of new information $\text{Var}(W_t)$ is diminishing as `fast' as possible, we take the \emph{alternative} extreme to be a `self-similarity' condition on $n$-way overlaps: where $\text{Cov}({X}_S, {X}_{S'}) = \rho \sqrt{\text{Var}({X}_S)\text{Var}({X}_{S'})}$ for any $S, S' \subset \mathcal{I}$. This then forces $W_t|Y \sim \mathcal{N}\left((1-\rho)^{t-1}\mu, (1-\rho)^{t-1}\sigma^2\right)$. We can then introduce a parameter $0 < \alpha \leq 1$ that determines the structure of the $n$-way overlap and interpolates between these two extremes. Intuitively, we can think of $\alpha$ as determining how quickly marginal experts run out of new information; small $\alpha$ means the marginal expert is adding substantial new information, while large $\alpha$ means the marginal expert's new information is quickly diminishing. Lastly, since $W_t$ captures information not revealed by any previous expert, we also have  $\text{Cov}(W_t, W_{t'}) = 0$ for all $t, t'$. Thus, we formally state our parameterization of the distribution of $W_t$.

\begin{assumption}[Expert residual signals as a conditional Gaussian process]
    Given outcome $Y$, the collection of expert residual signals $\left( W_t \right)_{t\geq 1}$ is a Gaussian process with mean $\pm f_{\alpha, t}\mu$ and covariance $\text{Cov}(W_t, W_{t'}) = \delta_{t, t'}f_{\alpha,t} \sigma^2$, where $f_{\alpha, 1} = 1$ and $f_{\alpha, t} = (1-\rho)^{1+(t-2)\alpha}$ for $t \geq 2$ and $0 < \alpha \leq 1$. For any collection of experts $S \subset \mathcal{I}$, the covariance matrix is $\Sigma_S = f_{\alpha, t}\sigma^2\mathbb{I}$.
\end{assumption}

\subsection{Signals to Log Likelihoods}
Now, we make an important observation (also noted by \citet{dasaratha2019aggregative}) that will make our model much easier to analyze: given our Gaussian assumptions on the signal distribution above, experts' posterior log odds are also Gaussian.\footnote{This follows from the assumption that the variance is the same when $Y=0$ and $Y=1$ (means can be arbitrarily different but re-centered and re-labeled to induce the same posterior belief). If we allow variances to differ, $\lambda_t$ has a generalized non-central chi-squared distribution with severe tractability issues.} 

\begin{observation}[Log likelihood ratios are normally distributed]
    For any signal $X$ distributed as $X|Y \sim \mathcal{N}\left( \pm \mu, \sigma^2\right)$ with $\mu > 0$, the log likelihood ratio of observing signal $X$ is $\lambda_t = \frac{2\mu}{\sigma^2}X$:
    {\small
\begin{equation}
    \begin{split}
    \lambda_{t} = \log \left( \frac{P(X_t|Y=1)}{P(X_t|Y=0)}\right) = \log \left( \frac{\exp\left( -\frac{(X_t - \mu)^2}{2\sigma^2}  \right)}{\exp\left( -\frac{(X_t + \mu)^2}{2\sigma^2} \right)}\right) =  \left( -\frac{(X_t - \mu)^2}{2\sigma^2} + \frac{(X_t + \mu)^2}{2\sigma^2}  \right) 
    =  \frac{2\mu}{\sigma^2}X 
    \end{split}
\end{equation}
}
Consequently, $\lambda_t|Y \sim \mathcal{N}\left(\pm \tau, 2\tau \right)$ where $\tau = \frac{2\mu^2}{\sigma^2}$.
\end{observation}

This observation will allow us to work exclusively with the posterior log odds transformed version of experts' signals, so we can talk directly about what experts \emph{believe} instead of what experts \emph{observe}. Indeed, going forward, we will only consider the log likelihood transformed versions of signals.

\begin{observation}[Log likelihood ratios of private signals]\label{obs:llpriv}
    Let $\lambda_t = \frac{2\mu}{\sigma^2} X_t$. Then $\left(\lambda_t\right)_{t\geq 1}$ is a Gaussian process adapted to the filtration $\mathbb{F}^{(0)}$, with mean $\pm \tau = \pm \frac{2\mu^2}{\sigma^2}$ and covariance $\text{Cov}(\lambda_t, \lambda_{t'}) = 2\tau (\rho + \delta_{t, t'}(1-\rho))$.
\end{observation}

\begin{observation}[Log likelihood ratios of residual signals]\label{obs:llresi}  
    Let $\psi_t = \frac{2\mu}{\sigma^2}W_t$. Then $\left(\psi_t\right)_{t \geq 1}$ is a Gaussian process adapted to the filtration $\mathbb{F}^{(1)}$, with mean $\pm f_{\alpha, t}\tau$ and covariance $\text{Cov}(\psi_t, \psi_{t'}) = 2\delta_{t,t'}f_{\alpha, t} \tau$.
\end{observation}

We can also define $\phi_t = \frac{2\mu}{\sigma^2}V_t$, so $\lambda_t = \phi_t + \psi_t$. Intuitively, we can interpret $\lambda_t$ as the log likelihood update to the prior induced by reading expert $t$'s rationale. Similarly, we can interpret $\psi_i$ as the log likelihood update to the prior induced by reading \emph{only the information in expert $t$'s rationale that no previous expert had written about}, and $\phi_t$ as the log likelihood update to the prior induced by reading all information in expert $t$'s rationale that had already been revealed by \emph{some} prior expert.

\section{Expert Aggregation} \label{sec:agg}

Having laid out the information each expert $t$ has at time-step $t$, we now turn our attention to the question of what the expert believes after observing all this information. After the expert observes their private signal $X_t$, their private posterior log odds is simply $\gamma = \lambda_t + \lambda_\pi$. However, at time-step $t$, expert $t$ also has access to $\mathcal{F}_t^{(0)}$ when no previous experts submitted rationales, and to $\mathcal{F}^{(1)}$ when all previous experts submitted rationales. We consider each of these in turn. In the former case, we focus on the adapted process $(\lambda_t)_{t\geq 1}$, and in the latter case, we will focus on the adapted process $(\psi_t)_{t\geq 1}$.

\subsection{Aggregation without rationales}

When no experts submit rationales, the expert at time-step $t$ only has access to $\mathcal{F}_t^{(0)} = \sigma\left(X_t | t \in \{1,\ldots, t\}\right)$. From Observation \ref{obs:llpriv}, this means that the expert can observe $\{\lambda_1, \ldots, \lambda_{t}\}$. Thus, the expert's aggregate belief given all information at time-step $t$ is $\mathbb{P}(Y=1|\mathcal{F}_t^{(0)})  = \mathbb{P}(Y=1 | \lambda_1, \ldots, \lambda_t)$. We give the following lemma to show how to compute this posterior log odds.

\begin{lemma}[Aggregation of Correlated Gaussian Signals]\label{lem:agg}
    Let $Y$ be a binary outcome with prior log odds $\lambda_\pi$ and $\Lambda$ be a vector of signals distributed as $\Lambda_t|Y=1 \sim \mathcal{N}(\mu, \Sigma)$ and $\Lambda|Y=0 \sim \mathcal{N}(-\mu, \Sigma)$. Then, the Bayesian posterior log odds that aggregates the signals is $\gamma = 2\mu^T \Sigma^{-1} \Lambda + \lambda_\pi$.
\end{lemma}
\begin{proof} The agent's posterior log odds are computed as follows:
    \begin{equation}
    \begin{split}
        \log \frac{P(Y= 1| \Lambda_t)}{P(Y= 0|\Lambda_t)} &= \log \frac{P(\Lambda_t | Y=1)}{P(\Lambda_t|Y=0)} + \log \frac{P(Y=1)}{P(Y=0)} \\
         &= \log \frac{\exp{\left( -\frac12 ({\Lambda_t} - {\mu})^T\Sigma^{-1}({\Lambda_t}-{\mu})\right)}}{\exp{\left( -\frac12 ({\Lambda_t} + {\mu})^T\Sigma^{-1}({\Lambda_t}+ {\mu})\right)}} + \lambda_\pi \\
         &=  -\frac12 ({\Lambda_t} - {\mu})^T\Sigma^{-1}({\Lambda_t}-{\mu})  +\frac12 ({\Lambda_t} + {\mu})^T\Sigma^{-1}({\Lambda_t}+ {\mu}) + \lambda_\pi \\
         &=  {\Lambda_t}^T\Sigma^{-1} {\mu} + {\mu}^T\Sigma^{-1}{\Lambda_t} + \lambda_\pi \\
         &=  2{\mu}^T\Sigma^{-1}{\Lambda_t} + \lambda_\pi\\
    \end{split}
\end{equation}
\end{proof}

With this lemma, we can give the following result on expert $t$'s posterior log odds when no experts report rationales, i.e., in the filtration $\mathbb{F}^{(0)}$:

\begin{theorem}[Aggregate without rationales]\label{lem:aggwout}
    By Observation \ref{obs:llpriv}, when no expert submits rationales, expert $t$ observes $\Lambda_{t} = \begin{pmatrix} {\lambda}_1, \ldots, {\lambda}_{t} \end{pmatrix}$ a vector of the log likelihood ratios of the private signals observed by experts $1, \ldots, t$ where $\Lambda_{t}|Y \sim \mathcal{N}\left(\pm \tau \mathds{1}_{t}, 2\tau((1-\rho)\mathbb{I}_t + \rho\mathds{1}_t\mathds{1}_t^T) \right)$. Then, expert $t$'s Bayesian posterior log odds is $\gamma_t =\lambda_\pi +  w\sum_{t'=1}^t \lambda_{t'}$ where $w = \frac{1}{1+(t-1)\rho}$. Ex ante, expert $t$'s posterior log odds is distributed as $\gamma_t|Y \sim \mathcal{N}\left(\lambda_\pi \pm tw\tau , 2tw\tau\right)$.
\end{theorem}

Observe that if $\rho=0$ and we had conditional independence, $\gamma = \lambda_\pi + \sum_{t'=1}^t \lambda_{t'}$, so we can just add up log posterior odds and arrive at the efficient aggregate. If $\rho =1$, then every signal is identical so $\lambda_t = \lambda_{t'}$ for all $1 \leq t' < t$, so the aggregate is naturally $\gamma_t = \lambda_\pi + \lambda_{t}$. In both cases, there is no need for rationales, since aggregating posterior log odds is already efficient. As we will show further on, it is when $0 < \rho < 1$ that rationales add value, since this is when it is helpful to disentangle overlapping pieces of information.

\subsection{Aggregation with rationales}

When every expert submits rationales, the expert at time-step $t$ has access to $\mathcal{F}_t^{(1)} = \sigma\left(X_S | S \subseteq \{1,\ldots, t\}\right)$. From Observation \ref{obs:llresi}, this means the expert can also observe the new information contributed by every expert $\{\psi_1, \ldots, \psi_{t}\}$. By Observation \ref{obs:resprop}.4, the expert's aggregate belief given all information at time-step $t$ is $\mathbb{P}(Y=1|\mathcal{F}_t^{(1)})  = \mathbb{P}(Y=1 | \psi_1, \ldots, \psi_t)$. We give the following result on expert $t$'s posterior log odds when all experts report rationales, i.e., in filtration $\mathbb{F}^{(1)}$:

\begin{theorem}[Aggregate with rationales]\label{thm:aggwith}
By Observation \ref{obs:llresi}, when every expert submits rationales, expert $t$ observes $\Psi_{t} = \begin{pmatrix} {\psi}_1, \ldots, {\psi}_{t} \end{pmatrix}$ a vector of the log likelihood ratios of the residual signals observed by experts $1, \ldots, t$ where $\Psi_{t}|Y \sim \mathcal{N}\left(\pm \tau \left(f_{\alpha, 1}, \ldots, f_{\alpha, t}\right), 2\tau \cdot \text{diag}\left(f_{\alpha, 1}, \ldots, f_{\alpha, t}\right)\right)$. Then, expert $t$'s Bayesian posterior log odds is $\gamma_t = \lambda_\pi + \sum_{t'=1}^t \psi_{t'}$. Ex ante, expert $t$'s posterior log odds is distributed as $\gamma_t|Y \sim \mathcal{N}\left(\lambda_\pi \pm \tau F_{\alpha,t}, 2\tau F_{\alpha,t}\right)$ where $F_{\alpha, 1}=1$ and for $t\geq 2$:
{\small
\begin{equation}
F_{\alpha,t} = \sum_{t'=1}^t f_{\alpha,t'} = 1 + (1-\rho)\left(\frac{1-(1-\rho)^{\alpha(t-1)}}{1-(1-\rho)^\alpha}\right)
\end{equation}
}
\end{theorem}

\subsection{Rationales enable more efficient aggregation}

We have thus far given the ex-ante posterior log odds at every time-step $t$ given $\mathcal{F}_t^{(0)}$ and $\mathcal{F}_t^{(1)}$. Now, we show that when all experts provide rationales, the ex-ante log posterior odds of expert $t$ is further away from the log prior odds, relative to when no experts provide rationales.

\begin{theorem}[Ex-ante log posterior odds is further from log prior odds with rationales]\label{thm:eff}
    Ex-ante, the expected log posterior odds at time-step $t$ is further from log prior odds when all experts provide rationales, relative to when no experts provide rationales, unless $\rho \in \{0, 1\}$. In other words, $tw\tau \leq \tau F_{\alpha, t}$ with equality when $\rho \in \{0, 1\}$.
\end{theorem}

Thus, we have established that rationales enable more efficient aggregation in expectation. While the expected log odds is easy to compute, the expected probability at a given time-step $\mathbb{E}[\sigma(\gamma_t)]$ where $\gamma_t$ is normally distributed does not have a simple closed-form expression. Instead, we provide a numerical visualization of ex-ante expected \emph{probabilities} with and without rationales, for various choices of $\rho, \alpha, \tau, \lambda_\pi$ as a function of the number of agents $t$ in Appendix \ref{app:visbel}. We find (as expected) that rationales are more valuable for smaller $\alpha$ (new experts add significant new information) and when $\rho$ is not too close to 0 or 1, while $\tau$ determines how quickly the expected belief gets stronger. The limiting behavior of the expected log posterior odds as $t \to \infty$ is tractable, which we give next. 

\begin{observation}[Limit of expected posterior without rationales] \label{obs:norat}
    Without rationales, the ex-ante expected log posterior odds converges when $\rho > 0$:
    {\small
    \begin{equation}
    \lim_{t\to \infty} \left(\lambda_\pi \pm \frac{t\tau}{1+(t-1)\rho} \right)= \lambda_\pi \pm \frac{\tau}{\rho}
    \end{equation}
    }
\end{observation}

\begin{observation}[Limit of expected posterior with rationales]\label{obs:withrat}
    With rationales, the ex-ante expected log posterior odds converges when $\rho > 0$:
    {\small
    \begin{equation}
        \lim_{t\to\infty} \left[ \lambda_\pi \pm \tau\left(1 + (1-\rho)\left(\frac{1-(1-\rho)^{\alpha(t-1)}}{1-(1-\rho)^\alpha}\right)\right) \right] = \lambda_\pi \pm \tau\left(1 + \frac{1-\rho}{1-(1-\rho)^\alpha}\right)
    \end{equation}
    }
    When $\alpha = 1$, $\lim_{t \to \infty} \tau F_{\alpha, t} = \lambda_\pi \pm \frac{\tau}{\rho}$.
\end{observation}

Observations \ref{obs:norat} and \ref{obs:withrat} quantify \emph{how much} better aggregation is ex-ante in the limit when all experts provide rationales, relative to when no rationales are provided. Even in the case where $\alpha=1$ and the rationales eventually provide only as much information as aggregates without rationales, rationales provide this faster (by Theorem \ref{thm:eff}). This is enabled by the fact that experts can efficiently discard previously revealed information, and focus on aggregating $(W_t)_{t\geq 1}$ without fear of `double-counting' information.

\section{Deliberation Mechanism}\label{sec:mechanism}

In the previous sections, we developed a model of reasoning with rationales that showed that rationales allow for more efficient aggregation. Now, we tackle the problem of providing incentives to elicit true rationales and beliefs, since providing rationales is typically a costly enterprise. We design a mechanism that a principal can use to \emph{elicit rationales from experts} and \emph{efficiently aggregate this information} to arrive at more accurate beliefs. We give our proposed mechanism $\mathcal{M}_{deliberation}$ below:

\begin{enumerate}
\item Experts $1, 2, \ldots$ arrive sequentially one at a time until the mechanism terminates at time-step $T$. Each expert makes a report $\kappa_t$ to the supervisor.
\item At time-step $t$, expert $t$ joins the mechanism with private information and observes the full history of reports $\kappa_1, \ldots, \kappa_{t-1}$ submitted to the supervisor as well as the supervisor's reported trajectory of beliefs $\hat{\gamma}_1^{S}, \ldots, \hat{\gamma}_{t-1}^{S}$. Expert $t$ then arrives at their posterior log odds $\gamma_t$.
\item Expert $t$ then chooses whether to exert effort ($e_t=0$ or $e_t=1$) to submit a rationale $\theta_t$, and whether to misreport their belief $\gamma_t$ and/or rationale $\theta_t$ (if applicable). The expert's report to the supervisor $\kappa_t$ consists of their belief, and if effort was exerted, their rationale.
\item The supervisor observes report $\kappa_t$ and updates their belief to $\gamma_t^{S}$. They then (strategically) report $\hat{\gamma}_t^{S}$ to the principal.
\item Upon realization of event $Y$, the principal determines payoffs for the supervisor and experts.
\end{enumerate}

\subsection{Background}

Before further analysis, we provide some key definitions. 
\begin{definition}[Scoring Rules]
A scoring rule is a function $S: Y \times \Delta_Y \rightarrow \mathbb{R}$ that scores a probabilistic prediction ${\bf p} \in \Delta_Y$ against an outcome $Y=y$. A scoring rule is \textit{proper} if $\mathbb{E}_{y \sim {\bf p}}[S(y, {\bf p})] \geq \mathbb{E}_{y \sim {\bf p}}[S(y, {\bf q})]$ for any ${\bf q} \neq {\bf p}, {\bf q} \in \Delta_Y$. The scoring rule is \textit{strictly proper} if the inequality is strict.
\end{definition}

\begin{definition}[Market Scoring Rules]
    An $S-$market scoring rule is a sequential and shared proper scoring rule that scores agent $t$ based on the difference in score between their report ${\bf p}_t$ and the prior report ${\bf p}_{t-1}$, i.e., $S_{MSR}(y, {\bf p}_t, {\bf p}_{t-1}) = S(y, {\bf p}_t) - S(y, {\bf p}_{t-1}) $ where $S$ is the given scoring rule. 
\end{definition}

\begin{remark}[Log Scoring Rule]
    The scoring rule $S(Y; {\bf p}) = \log(p_y)$ is a strictly proper scoring rule.
\end{remark}

\begin{definition}[Perfect Bayesian Equilibrium]
    A Perfect Bayesian Equilibrium (PBE) is a solution concept for sequential Bayesian games with incomplete information. It consists of a strategy profile—specifying each agent’s action at every possible information set—and a belief profile—assigning to each information set a belief updated by Bayes’ rule when applicable. In a PBE, no agent can improve their expected payoff by unilaterally deviating, and off-the-equilibrium-path beliefs (formed after unexpected moves) are required to be consistent with the equilibrium strategies. 
\end{definition}

\subsection{The Supervisor's Game}

Here, we lay out the game from the supervisor's perspective.

\paragraph{Reports and Strategy} At every time-step $t$, the supervisor employs strategy $s: \mathbb{R} \rightarrow \mathbb{R}$ to strategically report their true log odds belief $\gamma_t^S$ to the principal, i.e., the supervisor reports $\hat{\gamma}_t^S = s(\gamma_t^S)$. The supervisor's true belief $\gamma_t^S$ at time-step $t$ is a function of all expert reports $\kappa_1, \ldots, \kappa_t$ seen thus far.

\paragraph{Payoffs for Supervisor} When the mechanism terminates at time-step $T$, the principal will reward the supervisor's report $\hat{\gamma}_t^S$ at each time-step using the \emph{log scoring rule}, with the log score of the prior as the constant offset. Thus, the  supervisor is rewarded at every time-step based on improving the prior towards the true outcome $Y=y$:
{\small
\begin{equation}\label{eq:suprew}
    R_S\left(y, \{ \hat{\gamma}_{t}^{S} \}_{1:T} \right) = \sum_{t=1}^T  y\left[\log\left( \sigma(\hat{\gamma}_{t}^{S}) \right) - \log(\sigma(\pi))\right] + (1-y)\left[\log\left( \sigma(-\hat{\gamma}_{t}^{S}) \right) - \log(\sigma(-\pi))\right]
\end{equation}
}
where $\sigma(\cdot)$ is the logistic sigmoid, which transforms the supervisor's log odds belief to a probability that $Y=1$.

\paragraph{Utility} We model the supervisor's utility as simply their payoff from the principal when outcome $Y=y$:

\begin{equation}
    U_S\left(y; \{\hat{\gamma}_t^{S}\}_{1:T}\right) =  R_S\left(y; \{\hat{\gamma}_t^{S}\}_{1:T}\right)
\end{equation}

\subsection{The Experts' Game}\label{sec:exp}

Now, we consider the game from experts' perspective. 

\paragraph{Reports and Strategy} Expert $t$ makes two strategic choices: whether or not to exert effort $e_t=0$ or $e_t=1$ to submit the rationale, and how to report their true belief ${\gamma}_t$ and (if applicable) rationale ${\theta}_t$. Concretely, expert $t$ employs strategy $s_t: \mathbb{R} \times \mathcal{X} \rightarrow \mathbb{R} \times \mathcal{X} \times \{0, 1\}$ to strategically report their log odds belief $\gamma_t$ and (if applicable) rationale $\theta_t$ to the principal, i.e., the expert reports $\kappa_t = s_t(\gamma_t, \theta_t)$ where $\kappa_t = (\hat{\gamma}_t, \emptyset, 0)$ when $e_t=0$ and $\kappa_t = (\hat{\gamma}_t, \hat{\theta}_t, 1)$ when $e_t =1$. Experts may choose to expend effort and provide false, misleading, or incomplete information in the rationale $\hat{\theta}$. Such strategic misreporting of the rationale and belief is reflected in the decomposition $\hat{\gamma}_t = \hat{\psi}_t + \hat{\phi}_t$.\footnote{Since rationales induce filtration $\mathbb{F}^{(1)}$, the rationale $\theta_t$ fully specifies the belief $\gamma_t$. Thus, reported rationales $\hat{\theta}_t$ should be \emph{consistent} with reported beliefs $\hat{\gamma}_t$, in the sense that $\hat{\gamma}_t$ should reflect the belief held by someone with rationale $\hat{\theta}_t$. The expert does technically have the freedom to report an inconsistent rationale, but in our equilibrium analysis we say experts ignore such inconsistent reports.} Additionally, we note that our construction of the filtration $\mathbb{F}^{(1)}$ relied on experts reporting the log likelihood ratios of their observed signal $\lambda_t$, but our mechanism is collecting log posterior odds $\hat{\gamma}_t$. Happily, every expert can recover other experts' implied log likelihood ratios, preserving the applicable filtrations $\mathbb{F}^{(0)}$ and $\mathbb{F}^{(1)}$.

\begin{observation}
    If expert $t$ reports only their posterior log odds $\hat{\gamma}_t$ and no rationales, it is possible to infer the log likelihood ratio $\hat{\lambda}_t$ of their signal via the recurrence relation $\hat{\lambda}_{t} = (\hat{\gamma}_t - \lambda_\pi)(1+(t-1)\rho) - \sum_{t'=1}^{t-1} \hat{\lambda}_{t'}$. If  experts report both posterior log odds $\hat{\gamma}_t$ and rationales $\hat{\theta}_t$, it is possible to infer the log likelihood ratio of the residual signal $\psi_t = \hat{\gamma}_t  - \lambda_\pi - \sum_{t'=1}^{t-1} \psi_{t'}$. Thus, observing experts' posterior beliefs is equivalent to observing their log likelihood ratios when rationales are absent, or their residual signals when rationales are present.
\end{observation}

\paragraph{Payoffs for Experts} The principal rewards experts using the \emph{market scoring rule} based on the log score. Specifically, the principal rewards the expert based on the \emph{supervisor's} update.

{\small
\begin{equation}
    R_t(y, \kappa_t) = a_t \cdot \left(y\left[\log\left( \sigma(\hat{\gamma}_{t}^{S}) \right) - \log(\sigma(\hat{\gamma}_{t-1}^S))\right] + (1-y)\left[\log\left( \sigma(-\hat{\gamma}_{t}^{S}) \right) - \log(\sigma(-\hat{\gamma}_{t-1}^S))\right]\right)
\end{equation}
}
where $a_t$ is a `liquidity' parameter that scales the reward and $\sigma(\cdot)$ is the logistic sigmoid which transforms the supervisor's log odds belief to the probability that $Y=1$. In other words, expert $t$'s reward is determined based on \emph{how much they improved the supervisor's reported belief, relative to the supervisor's reported belief at $t-1$}.

\paragraph{Utility} We assume each expert has cost $c$ for exerting effort $e_t=1$; this is the cost for \emph{submitting the rationale} to the supervisor. Thus, expert $t$'s utility function given outcome $Y=y$ is:
\begin{equation}
    U_t(y, \kappa_t) = R_t(y, \kappa_t) - ce_t
\end{equation}

\subsection{Analyzing Incentives}

With the above setup, we can analyze the supervisor's and experts' incentives. We will proceed as follows: (1) we show that the supervisor strictly maximizes their utility by reporting their true beliefs to the principal at each time-step, i.e., $\hat{\gamma}_t^S = \gamma_t^S$; (2) we show that if the supervisor believes experts' reports $\kappa_1, \ldots, \kappa_t$ to be truthful (regardless of effort), the supervisor's true belief $\gamma_t^S$ at time-step $t$ will simply mirror expert $t$'s reported belief $\hat{\gamma}_t$; (3) when experts and supervisor are truthful, the supervisor's ex-ante \emph{utility} is greater if all experts report rationales, relative to when no experts report rationales; (4) this incentivizes the supervisor to \emph{commit} to a strategy of not updating their report at time-step $t$ if expert $t$ does not provide a rationale, i.e., $\hat{\gamma}_t^S = \hat{\gamma}_{t-1}^S$ when $e_t = 0$; and finally (5) when the supervisor commits to this strategy and principal scales the log score appropriately, it is a perfect Bayesian equilibrium for all experts to exert effort and submit both their belief and rationale, for the supervisor to report their true beliefs, and for all agents to believe all other agents are truthful, including off the equilibrium path. We begin with the first claim:

\begin{theorem}[Truthful reporting is strictly dominant strategy for supervisor]
    The supervisor's utility is strictly maximized when they report their own beliefs truthfully at each time-step, i.e., $\hat{\gamma}_t^{S} = \gamma_t^{S}$.
\end{theorem}

\begin{proof}
    The supervisor's total utility is the sum of payoffs at every every time-step; the payoff at every time-step is a linear transformation of the log scoring rule, which is strictly proper. Thus, the supervisor strictly maximizes their per-time-step utility as well as total utility by reporting their true belief, i.e., $\hat{\gamma}_t^S = \gamma_t^S$ for all $1 \leq t \leq T$. 
\end{proof}

This result shows us that supervisors are incentivized to report their true beliefs. But what do supervisors believe? We answer this next:

\begin{theorem}[Supervisor's true belief mimics truthful experts]\label{lem:mimic}
    If all expert reports $\kappa_1, \ldots, \kappa_t$ are truthful, then the supervisor's true belief at time-step $t$ mimics expert $t$'s reported belief, i.e., $\gamma_t^S = \hat{\gamma}_t = \gamma_t$.
\end{theorem}

\begin{proof}
    Since expert $t$ can observe all previous reports $\kappa_1, \ldots, \kappa_{t-1}$ and these reports are truthful, expert $t$'s updated belief $\gamma_t$ is the correct Bayesian aggregate of all revealed information thus far (regardless of whether or not rationales are revealed). The supervisor has no private information, and thus no basis to diverge from expert $t$'s belief. Since expert $t$ also reports their updated belief truthfully as $\hat{\gamma}_t$, the supervisor's updated belief at time-step $t$ simply mimics this, and $\gamma_t^S = \hat{\gamma}_t = \gamma_t$.
\end{proof}

We have established that the supervisor's true belief mimics experts' reported beliefs unless they have reason to disbelieve experts. What do experts actually report? In particular, do experts naturally report rationales to aid other experts' aggregation? The answer is no:

\begin{theorem}[Experts have no incentive to submit rationales]
    If cost of effort $c > 0$ and the expert believes the supervisor will report their true belief $\gamma_t^S$ to the principal, the expert's utility is strictly maximized by $e_t = 0$, i.e., not submitting the rationale.
\end{theorem}

\begin{proof}
    By Theorem \ref{lem:mimic}, the supervisor's true belief will simply mimic and report the expert's reported belief $\gamma_t^S = \hat{\gamma}_t$ regardless of whether expert $t$ provides a rationale (whether $e_t=0$ or $e_t=1$). Since there is a positive cost to submitting the rationale but the impact on the supervisor's belief and hence their own payoff is the same in both $e_t=0$ and $e_t=1$, the expert is strictly better off not submitting the rationale.
\end{proof}

Thus far, it appears that experts default to not submitting rationales, and the supervisor simply mimics the experts' reported belief and reports this truthfully to the principal at each time-step. What then is the role of the supervisor? Well, the supervisor's actual expected utility depends on the quality of the supervisor's beliefs, which in turn just depends on the quality of the experts' reports. We saw in Theorem \ref{thm:eff} that the ex-ante log posterior odds is greater in expectation when experts provide rationales; we now show that this translates into ex-ante \emph{higher expected score} at time-step $t$. Then, we can show that even though experts individually have no incentive to provide rationales, the fact that they collectively make the supervisor better off means the supervisor has incentive to coordinate experts to provide rationales. To arrive at this conclusion, we begin with the following result:

\begin{theorem}[Supervisor's Ex-Ante Utility is Strictly Greater with Rationales]\label{thm:great}
    Suppose the supervisor and all experts report their beliefs truthfully. Then, the supervisor's expected utility at every time-step is strictly greater when all experts report rationales, relative to when no experts report rationales, when $0 < \rho < 1$.
\end{theorem}

Thus, although the supervisor's true belief at each time-step will simply mimic the expert's reported (true) belief at that time-step, it matters to the supervisor's welfare whether the experts are providing rationales or not. Although experts ordinarily are individually strictly better off not providing rationales, the supervisor can induce all experts to report rationales \emph{by credibly committing to not reporting an updated belief to the principal if an expert does not report their rationale}. In other words, if expert $t$ does not submit a rationale, the supervisor commits to the strategy $s({\gamma}_t^S) = \hat{\gamma}_{t-1}^S$, i.e., $\hat{\gamma}_t^S = \hat{\gamma}_{t-1}^S$, so $U_t = R_t = 0$ when $e_t = 0$. Then, if the principal sets $a$ high enough to cover the expert's cost of effort to submit the rationale, the expected utility of submitting the rationale is positive, and the expected utility of not submitting the rationale is zero. Thus, the role of the supervisor becomes clear: their task is to commit to only accepting reports containing rationales, which will ex-ante improve their expected utility. These insights culminate in our main equilibrium result:

\begin{theorem}[Experts exerting effort and all agents reporting true beliefs is a PBE]\label{thm:main}
Suppose the supervisor commits to the strategy $s^*({\gamma}_t^S) = \hat{\gamma}_{t-1}^S$ if $e_t = 0$ and $s^*({\gamma}_t^S) = \gamma_t^S$ if $e_t = 1$. Let $s_t^*(\gamma_t, \theta_t) = (\gamma_t, \theta_t, 1)$ be the strategy where expert $t$ exerts effort and reports their belief and rationale truthfully. Let $b^*$ be the belief profile where the supervisor believes experts reports to be truthful and thus mimics them, and let $b^*_t$ be the belief profile where experts always believe other experts' reports to be truthful and update their log odds based on the information in filtration $\mathbb{F}^{(1)}$ according to Bayes' rule. Additionally, suppose the principal sets $a_t > \frac{c}{\mathbb{E}[S_{MSR}(Y; \sigma(\gamma_t), \sigma(\gamma_{t-1}))|e_{1:t}=1]}$. Then the strategy profile ${\bf s}^* = (s^*, s_1^*, \ldots, s_t^*)$ and the belief profile ${\bf b}^* = (b^*, b_1^*, \ldots, b_t^*)$ form a truthful and effortful perfect Bayesian equilibrium. The equilibrium is strictly truthful for supervisor's and experts' belief reports, and weakly truthful for experts' rationale reports.
\end{theorem}

\begin{proof}
    Since the supervisor commits to ignoring reports without rationales, the expert $t$'s expected utility $U_t=0$ when $e_t=0$. When $e_t=1$ and the expert reports their true belief and rationale, expert $t$'s expected utility $U_t = a \mathbb{E}[R_t(y, \kappa_t)] - c > 0$. Thus, exerting effort to report both belief and rationale is strictly higher expected utility than not exerting effort.
    
    If $e_t=1$ and the expert misreports their belief, the supervisor will believe this report and submit it to the principal; since the expert is paid with a proper scoring rule of the supervisor's report, the expert strictly maximizes their payoff only if $\hat{\gamma}_t^S = \gamma_t$. Thus, the given equilibrium is strictly truthful for experts' belief reports.
    
    If $e_t = 1$ and the expert truthfully reports their belief $\hat{\gamma}_t= \gamma_t$, they still have some freedom in how they their rationale. Specifically, they can report any $\theta_t'$ ($\hat{\theta}_t = \theta_t' \neq \theta_t$) that satisfies $\gamma_t = \phi'_t + \psi'_t$, i.e., the expert can report any rationale that is consistent with their true belief, and thus would induce the same update in the supervisor's belief. As we noted in Section \ref{sec:exp}, experts can technically report rationales inconsistent with their belief, i.e., report $\hat{\theta}_t$ such that the implied log likelihood ratios for old and new information do not match the experts' reported belief ($\hat{\gamma}_t \neq \hat{\phi}_t  + \hat{\psi}_t$). In the case of such \emph{nonsense reports}, we allow the supervisor and other agents ignore the report which again produces zero utility. Thus, the stated equilibrium is weakly truthful in the experts' rationale reports.
\end{proof}

Our equilibrium result above, as stated, is only weakly truthful for experts' rationale reports, since experts can report any rationale that induces their true belief $\gamma_t$ in the supervisor for the same cost $c$. However, we can make this equilibrium strict by additionally rewarding expert $t$ based on an appropriately scaled version of supervisor's \emph{future} log score -- then, only expert $t$'s \emph{true} rationale allows \emph{future} experts to efficiently aggregate information, and any deviation from truthfully revealing their rationale would `leave money on the table.'

\paragraph{Scaling the Payoff}
In Theorem \ref{thm:main}, we stated that the principal must scale the log score by $a_t > \frac{c}{\mathbb{E}[S_{MSR}(Y; \sigma(\gamma_t), \sigma(\gamma_{t-1}))|e_{1:t}=1]}$ to compensate expert $t$ for the cost of effort. Here,
{\small
\begin{equation}
\begin{split}
    \mathbb{E}[S_{MSR}(Y; \sigma(\gamma_t), \sigma(\gamma_{t-1}))|e_{1:t}=1] &= \mathbb{E}_{Y, \gamma_t|\gamma_{t-1}}\left[Y\log\left(\sigma\left(\gamma_t\right)\right) + (1-Y)\log\left(\sigma\left(-\gamma_t\right) \right)\right] \\
    &\phantom{=} - \mathbb{E}_{Y, \gamma_{t-1}}\left[Y\log\left(\sigma\left(\gamma_{t-1}\right)\right) + (1-Y)\log\left(\sigma\left(-\gamma_{t-1}\right)\right) \right] \\
    &= \mathbb{E}_{Y, \gamma_t|\gamma_{t-1}}\left[Y\log\left( \frac{\sigma\left(\gamma_{t-1} + \psi_t\right)}{\sigma\left(\gamma_{t-1}\right)}\right) + (1-Y)\log\left( \frac{\sigma\left(-\gamma_{t-1} - \psi_t\right)}{\sigma\left(-\gamma_{t-1}\right)}\right)\right] \\
    &= \sigma(\lambda_\pi)\mathbb{E}_{\gamma_{t-1}, \psi_t | Y=1}\left[\log\left( \frac{\sigma\left(\gamma_{t-1} + \psi_t\right)}{\sigma\left(\gamma_{t-1}\right)}\right)\right] \\
    &\phantom{=} + \sigma(-\lambda_\pi)\mathbb{E}_{\gamma_{t-1}, \psi_t | Y=0}\left[\log\left( \frac{\sigma\left(-\gamma_{t-1} - \psi_t\right)}{\sigma\left(-\gamma_{t-1}\right)}\right)\right]
    \end{split}
\end{equation}
}

When $e_{t'}=1$ for $1 \leq t' \leq t$ and belief and rationale reports are truthful, by Theorem \ref{thm:aggwith}, we have $\gamma_t|Y \sim \mathcal{N}\left(\lambda_\pi \pm \tau F_{\alpha,t}, 2\tau F_{\alpha,t}\right)$. Unfortunately, the expected log sigmoid of a normally distributed random variable does not have a closed form expression. Nevertheless, we can compute this numerically for any time-step $t$ given parameters $\alpha$ (rate at which new information diminishes), $\rho$ (overlap in experts' information), $\tau$ (absolute amount of information held by each agent), $\lambda_\pi$ (log prior odds). We visualize experts' expected scores in the truthful PBE in Figure \ref{fig:expscore}; inverting this and scaling by the cost $c$ thus determines the principal's parameter setting $a_t$.

\section{Alternative Deliberation Mechanisms} \label{sec:alt}

The mechanism we presented in Section \ref{sec:mechanism} builds naturally on the model in Section \ref{sec:model} and aggregation analysis in Section \ref{sec:agg}. However, the same model and analysis can be used to justify other designs and approaches as well, which we discuss here.

\paragraph{Fusing Supervisor and Principal}

A natural variant of our mechanism merges the roles of supervisor and principal into a single agent. In settings where the principal has enough time or resources to read every rationale directly, this fused role simplifies our mechanism and eliminates the need for an intermediary supervisor. For this to work, the principal must credibly commit to two things: (a) not updating their reported belief in the absence of a rationale (which incentivizes experts to provide a rationale) and (b) always reporting their own updated belief truthfully once a rationale is provided (i.e., there is no need to incentivize the principal-supervisor). When the principal’s time is too valuable to read a large number of rationales, or if there are concerns about whether the principal can credibly commit to ignoring reports without rationales, introducing a dedicated, incentivized supervisor may be beneficial, as we do in our preferred design.

\paragraph{Deliberation Markets}

We can also cast our deliberation mechanism as a `market for rationales' analogous to a prediction market. This follows from the result that reports scored using a market scoring rule can be equivalently cast as a trade against an automated market maker implementing a proper scoring rule \citep{chen2012utility}. A deliberation market could be thought of as a prediction market with exactly two participants: the supervisor and the principal. The initial price is determined by the prior $\sigma(\lambda_\pi)$, and the supervisor trades \emph{on behalf of the expert} against the principal, who operates an automated market maker willing to take either side of the trade. Specifically, at every time-step $t$, the supervisor observes expert $t$'s report, arrives at an updated belief, and executes a trade against the principal in the prediction market reflecting the updated belief \emph{if the expert provided a rationale}. The liquidity at every time-step is determined by the liquidity parameter $a_t$ from experts' rewards. Separately, the principal computes the supervisor's payout by simulating the supervisor's trades against a fresh prediction market with initial price $\sigma(\lambda_\pi)$ at every time-step. Thus, a deliberation market is simply a prediction market mechanism where the supervisor trades against the principal at every time-step based on the rationale supplied by the expert. An interesting extension would be to allow experts to `forward' funds along with their report that the supervisor uses to trade on their behalf; then, the decision on how much liquidity to consume is determined by experts based on their risk preferences. However, while the expert provides the funds, the ultimate decision of what trade to execute is made by the supervisor based on the rationale; importantly, the supervisor makes no trade if expert $t$ submits no rationale. This is interestingly reminiscent of ideas like proxy voting and liquid democracy \citep{miller1969program, kahng2021liquid}.

\paragraph{Self-Resolving Deliberation Mechanism}

While our proposed mechanism relies on the realization of the ground truth to determine payoffs, we can extend our mechanism to elicit beliefs and rationales \emph{even when the ground truth is unavailable}. This variant fuses the principal and supervisor (as discussed above) and adapts the design of \emph{self-resolving prediction markets} \citep{srinivasan2023selfresolving}, which elicits truthful beliefs in the absence of ground truth. In a self-resolving prediction market, agents report beliefs sequentially, the mechanism terminates with constant probability after every report, and agents are scored against the final report at termination. In a self-resolving deliberation mechanism, experts can still observe all previous reports in addition to their own private information, and report their belief and (optionally) their rationale directly to the principal. The principal then records their updated belief after every report, but commits to not recording an updated belief at time-step $t$ if expert $t$ does not submit a rationale. The mechanism terminates with constant (small) probability after the principal records their belief and the principal's belief at termination proxies for the ground truth. Thus, experts are rewarded based on \emph{how much they moved the principal toward their eventual belief}. As in self-resolving prediction markets, if the termination probability is small enough, the principal's final recorded belief will be sufficiently distant in expectation from an expert's belief that they have little hope of reliably misleading the principal. Additionally, as a practical consideration, when experts submit rationales, it should be even more challenging for any misleading rationales to persist, since subsequent agents can provide corrections.

\paragraph{Deliberation Trees}
In large-scale settings where a single supervisor cannot process all rationales, we can adapt our deliberation mechanism into a tree structure. Multiple small groups of experts gather to sequentially submit reports to a local supervisor, and each mini-deliberation mechanism serves as a node that produces a single, `locally aggregated’ belief and rationale. Each local supervisor then submits their aggregate belief and an aggregated rationale onward to a higher-level supervisor who sequentially receives reports from lower-level supervisors from multiple nodes. This higher-level supervisor similarly operates within a local deliberation mechanism, sequentially aggregating intermediate beliefs and passing an intermediate aggregated rationale onward to yet another higher-level supervisor. At each level of the tree, the higher-level supervisor commits to not updating their belief if they observe a report without a rationale, which maintains the incentive for lower-level supervisors (or experts, at the lowest level) to submit aggregated rationales. Ultimately, the root node (the final supervisor) receives a set of local rationales that themselves aggregate all relevant information instead of a larger number of original expert submissions. Since every belief report is ultimately scored with a proper scoring rule, we maintain strict truthfulness for reporting beliefs. This hierarchical approach---forming a `deliberation tree’---keeps the job of reading rationales manageable at every node, and may be important if the supervisor's role incurs a cost, especially if this cost increases non-linearly with the number of rationales they process. 
This variant is useful when it is necessary to relieve the supervisor's burden.

\paragraph{Incorporating LLM Agents} Given that our mechanism is focused on eliciting natural-language rationales, we consider how LLMs can incorporated into our mechanism. The most natural fit is to use \emph{LLMs as supervisors}: the supervisor's primary role is to read experts' rationales and arrive at an updated belief commensurate with the new information revealed in a rationale. The job of analyzing and summarizing textual information plays to the strengths of LLMs, and using LLMs to reason about information and generate forecasts is an active area of research \citep{halawi2024approaching, schoenegger2024wisdom, karger2024forecastbench}. An LLM supervisor could also distill an aggregate rationale from the supplied expert rationales so future experts have a single reference they can observe to  understand what information has already been revealed. Additionally, an LLM supervisor would not have to be incentivized, allowing a more practical construction of self-resolving deliberation mechanisms. As LLM capabilities advance, we may also consider using LLMs as experts where different LLMs could contribute general or domain-specific insights to a centralized supervisor. Our proposed deliberation mechanisms could then be used to compensate such LLM agents for the information they provide based on its value to the supervisor. 

\section{Conclusion and Future Work}

This work makes two contributions: (1) we provide a simple model of how agents with overlapping information can reveal explanations or \emph{rationales} that enable more efficient information aggregation, and (2) we use this model to develop a mechanism for eliciting costly rationales in which reporting true beliefs and rationales is a perfect Bayesian equilibrium. Our proposed model captures the commonsense idea that hearing an explanation for a belief is more valuable than simply hearing the belief, as it allows the listener to discount information they already knew to arrive at a more efficient aggregate than they could by naively aggregating beliefs. Based on this model, we proposed a simple sequential mechanism to truthfully elicit rationales and aggregated beliefs from experts, using a supervisor who reviews experts' reports for new information before accepting the report. Our mechanism has natural applications to judgmental forecasting and crowdsourcing as it provides incentives to obtain rationales in addition to agents' beliefs.  Directions for future work that are of particular interest include exploring whether such mechanisms could: enable human-LLM hybrid forecasting setups; provide rewards/loss signals for AI alignment (particularly the self-resolving variant which does not require ground truth); serve as an crowdsourcing mechanism for eliciting reasoning chains with applications to training LLMs; be used for applications in community-based content moderation;  In particular, we identify experimental validation and LLM-focused applications are natural next steps for future work.

\subsection*{Acknowledgments}
We thank Kevin He for helpful discussions. This work was supported by the OpenAI Superalignment Fast Grant.

\bibliographystyle{apalike}
\bibliography{winnower_template}

\appendix
\section{Proofs}

\textbf{Theorem \ref{lem:aggwout}}
    \emph{By Observation \ref{obs:llpriv}, when no expert submits rationales, expert $t$ observes $\Lambda_{t} = \begin{pmatrix} {\lambda}_1, \ldots, {\lambda}_{t} \end{pmatrix}$ a vector of the log likelihood ratios of the private signals observed by experts $1, \ldots, t$ where $\Lambda_{t}|Y \sim \mathcal{N}\left(\pm \tau \mathds{1}_{t}, 2\tau((1-\rho)\mathbb{I}_t + \rho\mathds{1}_t\mathds{1}_t^T) \right)$. Then, expert $t$'s Bayesian posterior log odds is $\gamma_t =\lambda_\pi +  w\sum_{t'=1}^t \lambda_{t'}$ where $w = \frac{1}{1+(t-1)\rho}$. Ex ante, expert $t$'s posterior log odds is distributed as $\gamma_t|Y \sim \mathcal{N}\left(\lambda_\pi \pm tw\tau , 2tw\tau\right)$.}

\begin{proof}
The inverse of the covariance matrix is $\Sigma^{-1} = \frac{1}{2\tau(1 - \rho)} \left( \mathbb{I} - \frac{\rho}{1 - \rho + \rho n} \, \mathds{1}\mathds{1}^T \right)$. Thus, by Lemma \ref{lem:agg}, the expert's posterior log odds is:

{\small
\begin{equation}
    \begin{split}
        \gamma &= \log \frac{P(Y= 1| \Lambda_t)}{P(Y= 0|\Lambda_t)} \\
        &= 2{\mu}^T\Sigma^{-1}{\Lambda_t} + \lambda_\pi\\
        &= 2\left(\tau \mathds{1}^T\right)\left( \frac{1}{2\tau(1 - \rho)} \left( \mathbb{I} - \frac{\rho}{1 - \rho + \rho t} \, \mathds{1}\mathds{1}^T \right)\right) + \lambda_\pi\\
        &= \left( \frac{1}{1 - \rho} \left( \mathds{1}^T - \frac{\rho}{1 - \rho + \rho t} \, \mathds{1}^T\mathds{1}\mathds{1}^T \right)\right) + \lambda_\pi\\
        &= \left( \frac{1}{1 - \rho} \left( 1- \frac{\rho t}{1 - \rho + \rho t}  \right)\mathds{1}^T\right) + \lambda_\pi\\
        &= \frac{1}{1 + (t-1)\rho}  \mathds{1}^T \Lambda_t + \lambda_\pi\\
        &= \lambda_\pi + \frac{1}{1 + (t-1)\rho} \sum_{i=1}^t \lambda_t
    \end{split}
\end{equation}
}
Writing $w = \frac{1}{1 + (t-1)\rho}$, we can compute the distribution of $w\sum_{t'=1}^t \lambda_{t'}$. Since each $\lambda_{t'}$ is normally distributed, the sum is also normally distributed.  The mean is computed as follows:
\begin{equation}
        \mathds{E}\left[w\sum_{t'=1}^t \lambda_{t'}\right] = w\sum_{t'=1}^t \mathds{E}\left[\lambda_{t'}\right] = \pm tw\tau
\end{equation}
The variance is computed as follows:
\begin{equation}
\begin{split}
    \text{Var}\left(w\sum_{t'=1}^t \lambda_{t'}\right) &= \sum_{t'=1}^t w^2 \text{Var}\left(\lambda_{t'}\right) + 2 w^2 \sum_{t'=1}^t\sum_{i=1}^{t'} \text{Cov}\left(\lambda_{t'}, \lambda_{i}\right) \\
    &=   2tw^2\tau + 2 w^2  \frac{t(t-1)}{2} \cdot 2\tau\rho \\
    &=   2tw^2\tau (1 +  (t-1)\rho) \\
    &=   2tw\tau \\
\end{split}
\end{equation}

Thus, ex ante, expert $t$'s posterior log odds formed by aggregating their private information with previous agents' reports is distributed as:
\begin{equation}
    \gamma \sim \mathcal{N}\left( \lambda_\pi  \pm tw\tau , 2tw\tau\right)
\end{equation}

\end{proof}

\noindent\textbf{Theorem \ref{thm:aggwith}}
\emph{By Observation \ref{obs:llresi}, when every expert submits rationales, expert $t$ observes $\Psi_{t} = \begin{pmatrix} {\psi}_1, \ldots, {\psi}_{t} \end{pmatrix}$ a vector of the log likelihood ratios of the residual signals observed by experts $1, \ldots, t$ where}
{\small
\begin{equation}
\Psi_{t}|Y \sim \mathcal{N}\left(\pm \tau \left(f_{\alpha, 1}, \ldots, f_{\alpha, t}\right), 2\tau \cdot \text{diag}\left(f_{\alpha, 1}, \ldots, f_{\alpha, t}\right)\right)
\end{equation}
}
\emph{Then, expert $t$'s Bayesian posterior log odds is $\gamma_t = \lambda_\pi + \sum_{t'=1}^t \psi_{t'}$. Ex ante, expert $t$'s posterior log odds is distributed as $\gamma_t|Y \sim \mathcal{N}\left(\lambda_\pi \pm \tau F_{\alpha,t}, 2\tau F_{\alpha,t}\right)$ where $F_{\alpha, 1}=1$ and for $t\geq 2$:}
{\small
\begin{equation}
F_{\alpha,t} = \sum_{t'=1}^t f_{\alpha,t'} = 1 + (1-\rho)\left(\frac{1-(1-\rho)^{\alpha(t-1)}}{1-(1-\rho)^\alpha}\right)
\end{equation}
}

\begin{proof} The inverse of the covariance matrix is simply $\Sigma^{-1} = \text{diag}\left(\frac{1}{2f_{\alpha, 1}\tau}, \ldots, \frac{1}{2f_{\alpha, t}\tau}\right)$. Thus, expert $t$'s aggregate is simply:

{\small
\begin{equation}
\begin{split}
    \gamma = \log \frac{P(Y= 1| \Psi_t)}{P(Y= 0|\Psi_t)} + \lambda_\pi = 2{\mu}^T\Sigma^{-1}{\Psi_t} + \lambda_\pi  = \lambda_\pi + 
    \sum_{t'=1}^t \psi_{t'}
    \end{split}
\end{equation}
}

Since the $\psi_t$ are normally distributed, the aggregate $\gamma_t$ is also normally distributed. Since the $\psi_t$ are conditionally independent, we can simply add the means and variances to get the distribution of expert $t$'s ex ante aggregate belief:
{\small
\begin{equation}
    \gamma_t | Y \sim \mathcal{N}\left(\lambda_\pi \pm \tau F_{\alpha,t}, 2\tau F_{\alpha,t} \right)
\end{equation}
}
where 
{\small
\begin{equation}
    \begin{split}
        F_{\alpha,t} &= \sum_{t'=1}^t f_{\alpha, t'} \\
        &= 1 + \sum_{t'=2}^t (1-\rho)^{1+(t'-2)\alpha} \\
        &= 1 + (1-\rho)\left(\sum_{t'=0}^{t-2} (1-\rho)^{t'\alpha}\right) \\
        &= 1 + (1-\rho)\left(\frac{1-(1-\rho)^{\alpha(t-1)}}{1-(1-\rho)^\alpha}\right)
    \end{split}
\end{equation}
}
\end{proof}

\noindent\textbf{Theorem  \ref{thm:eff}}
    \emph{Ex-ante, the expected log posterior odds at time-step $t$ is further from log prior odds when all experts provide rationales, relative to when no experts provide rationales, unless $\rho \in \{0, 1\}$. In other words, $tw\tau \leq \tau F_{\alpha, t}$ with equality when $\rho \in \{0, 1\}$.}

\begin{proof}
 We can manipulate $F_{\alpha, t}$ as follows:

    \begin{equation}
    \begin{split}
        F_{\alpha, t} &= 1 + (1-\rho)\left(\sum_{t'=0}^{t-2} (1-\rho)^{t'\alpha}\right) \\
        &\geq 1 + (1-\rho)\left(\sum_{t'=0}^{t-2} (1-\rho)^{t'}\right) \\
        &= \sum_{t'=0}^{t-1} (1-\rho)^{t'} \\
        &= \frac{1}{\rho}(1 - (1-\rho)^t) \\
        &= \frac{1}{\rho\left[ 1+(t-1)\rho\right]}(1 - (1-\rho)^t)\left( 1+(t-1)\rho\right) \\
        &= \frac{1}{\rho\left[ 1+(t-1)\rho\right]}\left( 1 + \rho t - \rho - (1-\rho)^t\left( 1+(t-1)\rho\right) \right) \\
        &= \frac{1}{\rho\left[ 1+(t-1)\rho\right]}\left( (1-\rho)\left[1 - (1-\rho)^{t-1}(1 + (t-1)\rho)\right] +\rho t\right) \\
        &= \frac{1-\rho}{\rho\left( 1+(t-1)\rho\right)} \left[1 - (1-\rho)^{t-1}(1 + (t-1)\rho) \right] + \frac{t}{ 1+(t-1)\rho}\\
        &= \frac{ 1-\rho}{\rho\left( 1+(t-1)\rho\right)}Q  + \frac{t}{ 1+(t-1)\rho}\\
    \end{split}
    \end{equation}

where $Q = \left[1 - (1-\rho)^{t-1}(1 + (t-1)\rho)\right] = 1-Q'$ and $Q' = (1-\rho)^{t-1}(1 + (t-1)\rho)$. Then, to upper bound $Q'$, we log transform it and show that this monotonically decreases from $t=1$:
\begin{equation}
    \log(Q') = \log\left((1-\rho)^{t-1}(1 + (t-1)\rho)\right) = (t-1)\log(1-\rho) + \log(1 + (t-1)\rho)
\end{equation}
At $t=1$, we have that $\log(Q') = 0$. Taking the first derivative, we see:
\begin{equation}
    \begin{split}
        \frac{d}{dt} \log(Q') = \log(1-\rho) + \frac{\rho}{1+\rho(t-1)}
    \end{split}
\end{equation}

Since $\log(1-\rho) \leq -\rho$ and $\frac{\rho}{1+\rho(t-1)} \leq \rho$, $\frac{d}{dt}\log(Q') \leq -\rho + \rho = 0$. Thus, we have $\log(Q') \leq 0 \Rightarrow Q' \leq 1 \Rightarrow Q = 1- Q' \geq 0$. Consequently,

    \begin{equation}
        \begin{split}
    F_{\alpha, t} &\geq \frac{ 1-\rho}{\rho\left( 1+(t-1)\rho\right)}Q  + \frac{t}{ 1+(t-1)\rho}\\
        &\geq \frac{t}{ 1+(t-1)\rho} \\
        &= tw
    \end{split}
    \end{equation}

    We have equality when $\rho = 0$ or $\rho = 1$. Since $\tau > 0$, we have shown that $tw\tau < \tau F_{\alpha, t}$. 
\end{proof}

\noindent\textbf{Theorem \ref{thm:great}}
    \emph{Suppose the supervisor and all experts report their beliefs truthfully. Then, the supervisor's expected utility at every time-step is strictly greater when all experts report rationales, relative to when no experts report rationales, when $0 < \rho < 1$.}

\begin{proof} We saw in Theorem \ref{thm:eff} that at each time-step, ex-ante, the expert's posterior log odds drifted further from the prior log odds in expectation when all experts report rationales, relative to when no experts report rationales. Here, we show that this translates to a higher expected score at each time-step, which in turn translates to higher expected utility at each time-step. The supervisor's expected score at time-step $t$ is:

\begin{equation}
    \begin{split}
        \mathds{E}_Y\left[\mathds{E}_{\gamma_t|Y}\left[S\left(Y; \sigma\left(\gamma_t\right)\right) \right]\right] &= \pi \cdot \mathds{E}_{\gamma_t|Y=1}\left[S\left(Y=1; \sigma\left(\gamma_t\right)\right) \right] + (1-\pi) \cdot \mathds{E}_{\gamma_t|Y=0}\left[S\left(Y=0; \sigma\left(\gamma_t\right)\right) \right]
    \end{split}
\end{equation}

Although the mean log posterior odds with rationales is greater than without rationales (by Theorem \ref{thm:eff}), so is the variance. Thus, we need a little more work to show that this still translates to a higher score.  We begin by showing that when $\gamma_t \sim \mathcal{N}(\lambda_\pi + {\xi, 2\xi})$, the expected log score under $Y=1$ is increasing in $\xi$. We begin by reparameterizing with $\gamma_t = \lambda_\pi +  \xi + \sqrt{2\xi}Z$ where $Z \sim \mathcal{N}(0, 1)$ and writing $f_\xi(Z) = \log\left(\sigma\left(\lambda_\pi + \xi + \sqrt{2\xi}Z\right)\right)$. Then, we want to show that:

\begin{equation}
    \frac{d}{d\xi} \mathbb{E}_Z\left[f_\xi(Z)\right] > 0
\end{equation}

Since $\left|\frac{\partial f}{\partial \xi}\right| = \left|\sigma(-(\lambda_\pi + \xi + \sqrt{2\xi}Z))\left(1 + \frac{Z}{\sqrt{2\xi}}\right)\right| \leq 1 + \left|\frac{Z}{\sqrt{2\xi}}\right|$ which is integrable under the probability measure of $Z$, by the dominated convergence theorem we can move the derivative inside:
{\small
\begin{equation}
\begin{split}
    \frac{d}{d\xi} \mathbb{E}_Z\left[f_\xi(Z)\right] &=  \mathbb{E}_Z\left[\frac{d}{d\xi} f_\xi(Z)\right] \\
    &=  \mathbb{E}_Z\left[f'(\lambda_\pi + \xi + \sqrt{2\xi}Z) \left(1  +\frac{Z}{\sqrt{2\xi}}\right) \right] \\
    &=  \mathbb{E}_Z\left[f'(\lambda_\pi + \xi + \sqrt{2\xi}Z)\right] + \frac{1}{\sqrt{2\xi}}\mathbb{E}_Z\left[f'(\lambda_\pi + \xi + \sqrt{2\xi}Z) Z \right]
    \end{split}
\end{equation}
}

Then, by Stein's lemma, we can write $\mathbb{E}_Z\left[f'(\lambda_\pi + \xi + \sqrt{2\xi}Z) Z \right] = \sqrt{2\xi}\mathbb{E}_Z\left[f''(\lambda_\pi + \xi + \sqrt{2\xi}Z)  \right]$:
{\small
\begin{equation}
    \frac{d}{d\xi} \mathbb{E}_Z\left[f_\xi(Z)\right] =  \mathbb{E}_Z\left[f'(\lambda_\pi + \xi + \sqrt{2\xi}Z) + f''(\lambda_\pi + \xi + \sqrt{2\xi}Z) \right]
\end{equation}
}
Since $f_\xi(Z) = \log\left(\sigma(\lambda_\pi + \xi + \sqrt{2\xi}Z)\right) = -\log\left(1 +\exp{\left(-\lambda_\pi 
-\xi-\sqrt{2\xi}Z\right)}\right)$, for $f(x)$
{\small
\begin{equation}
    \begin{split}
        \mathbb{E}_Z\left[f'(\lambda_\pi + \xi + \sqrt{2\xi}Z) + f''(\lambda_\pi + \xi + \sqrt{2\xi}Z) \right] &= \mathbb{E}_Z\left[\left. \frac{\partial f}{\partial x} \right|_{x = \lambda_\pi + \xi + \sqrt{2\xi}Z} + \left. \frac{\partial^2 f}{\partial (x)^2} \right|_{x = \lambda_\pi + \xi + \sqrt{2\xi}Z}\right] \\
        &= \mathbb{E}_Z\left[\frac{1}{1 + \exp\left(\lambda_\pi + \xi + \sqrt{2\xi}Z\right)} - \frac{\exp(\lambda_\pi + \xi + \sqrt{2\xi}Z)}{\left( 1 + \exp{(\lambda_\pi + \xi + \sqrt{2\xi}Z)}\right)^2}\right] \\
        &=  \mathbb{E}_Z\left[\frac{1}{\left( 1 + \exp{(\lambda_\pi + \xi + \sqrt{2\xi}Z)}\right)^2}\right] \\
        &>0
    \end{split}
\end{equation}
}

for any value of $Z$. Since the gradient of the expectation is strictly positive, the expected log score is increasing in $\xi$. By a symmetric argument, we have that when $\gamma_t \sim \mathcal{N}(\lambda_\pi - {\xi, 2\xi})$, the expected log score under $Y=0$ is increasing in $\xi$ (recall $\xi > 0$). Since the expected log score under $Y=0$ and $Y=1$ is increasing in $\xi$, the unconditional expected log score is also increasing in $\xi$. By Theorem \ref{thm:eff}, $\tau F_{\alpha, t} > tw\tau$ when $0 < \rho < 1$; thus we shown our claim.
\end{proof}

\subsection{The Value of Rationales}\label{app:vis}

\subsubsection{Visualizing Expected Beliefs}\label{app:visbel}

Here, we graph how the aggregated belief evolves in expectation with the number of experts. We take $\lambda_\pi=0$.

\begin{figure}
\begin{center}
    \includegraphics[width=0.9\linewidth]{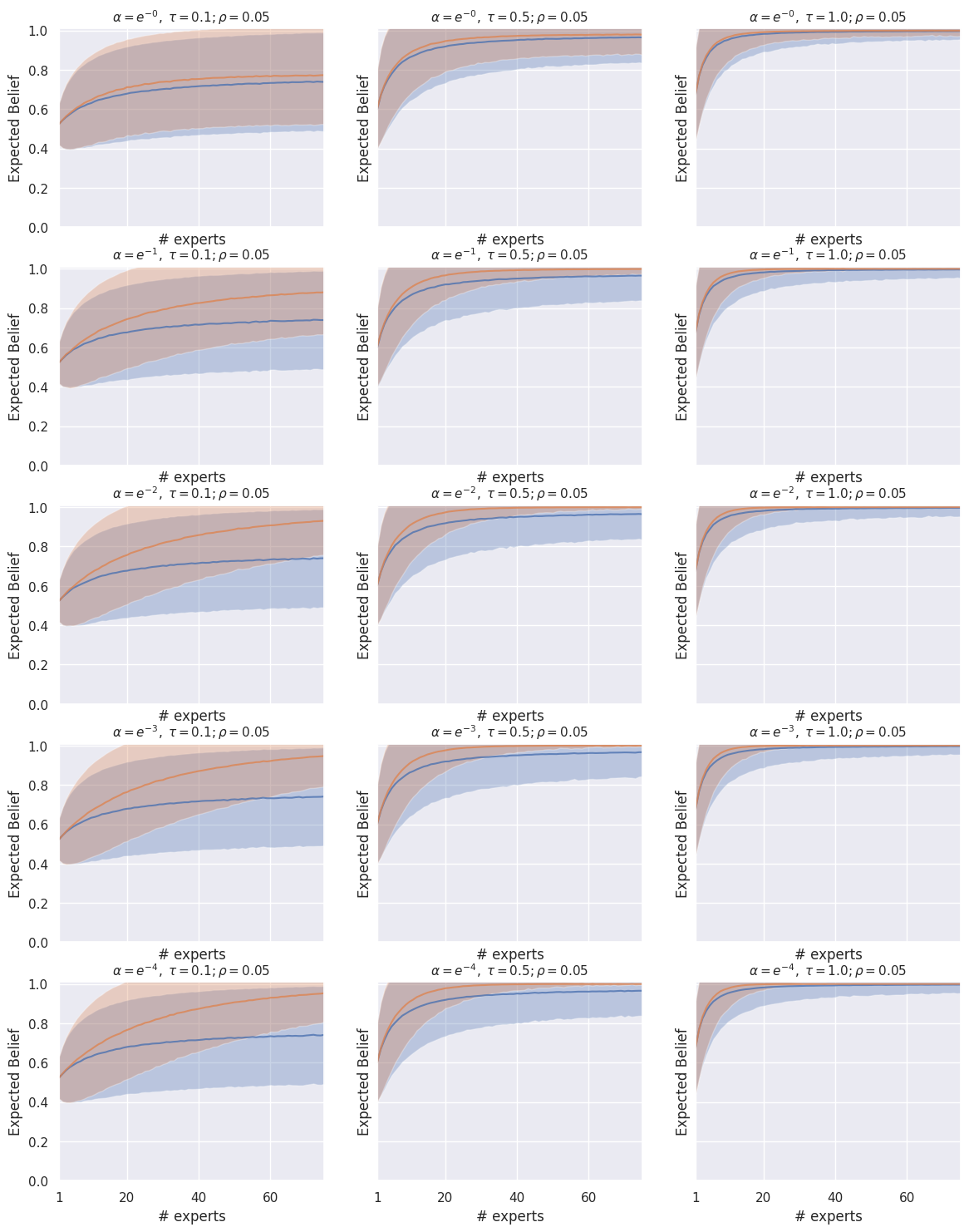}
    \caption{Expected belief (probability) as a function of the number of experts given rationales (orange) and without rationales (blue) when $\rho=0.05$ and $Y=1$, for various values of $\alpha$, $\tau$. Shaded area is 1 std deviation.}
\end{center}
\end{figure}

\begin{figure}
\begin{center}
    \includegraphics[width=0.9\linewidth]{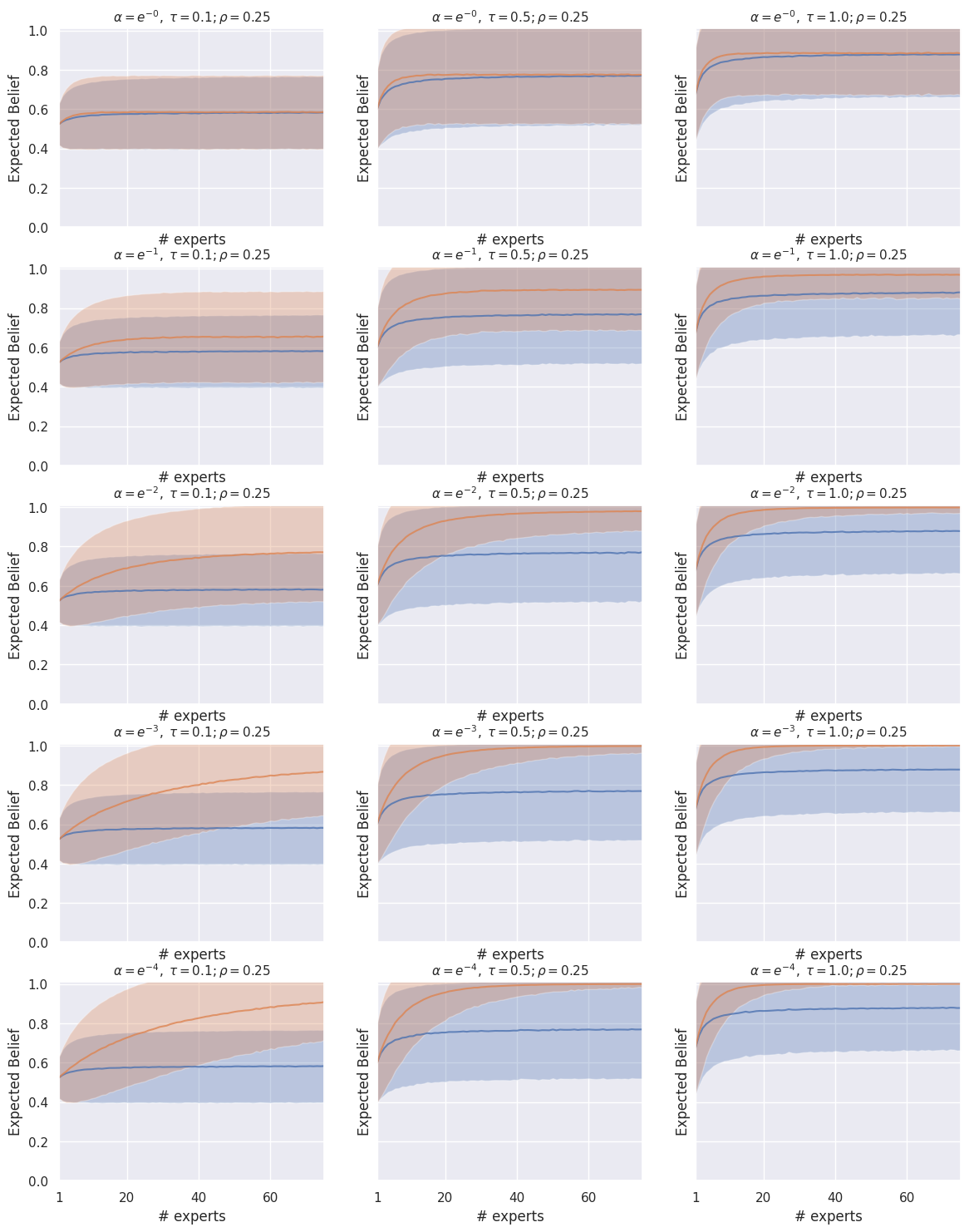}
    \caption{Expected belief (probability) as a function of the number of experts given rationales (orange) and without rationales (blue) when $\rho=0.25$ and $Y=1$, for various values of $\alpha$, $\tau$. Shaded area is 1 std deviation.}
\end{center}
\end{figure}

\begin{figure}
\begin{center}
    \includegraphics[width=0.9\linewidth]{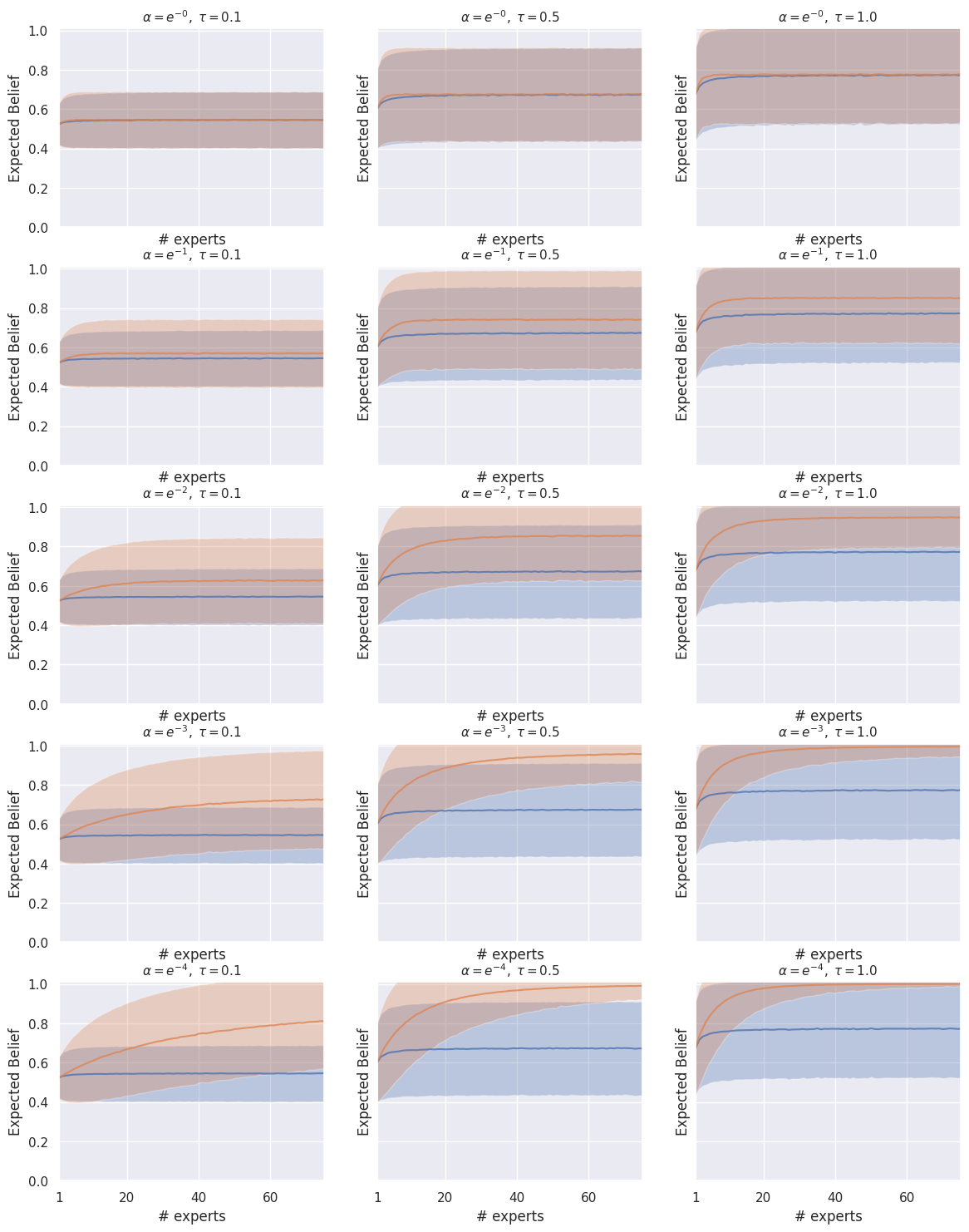}
    \caption{Expected belief (probability) as a function of the number of experts given rationales (orange) and without rationales (blue) when $\rho=0.5$ and  $Y=1$, for various values of $\alpha$, $\tau$. Shaded area is 1 std deviation.}
\end{center}
\end{figure}

\begin{figure}
\begin{center}
    \includegraphics[width=0.9\linewidth]{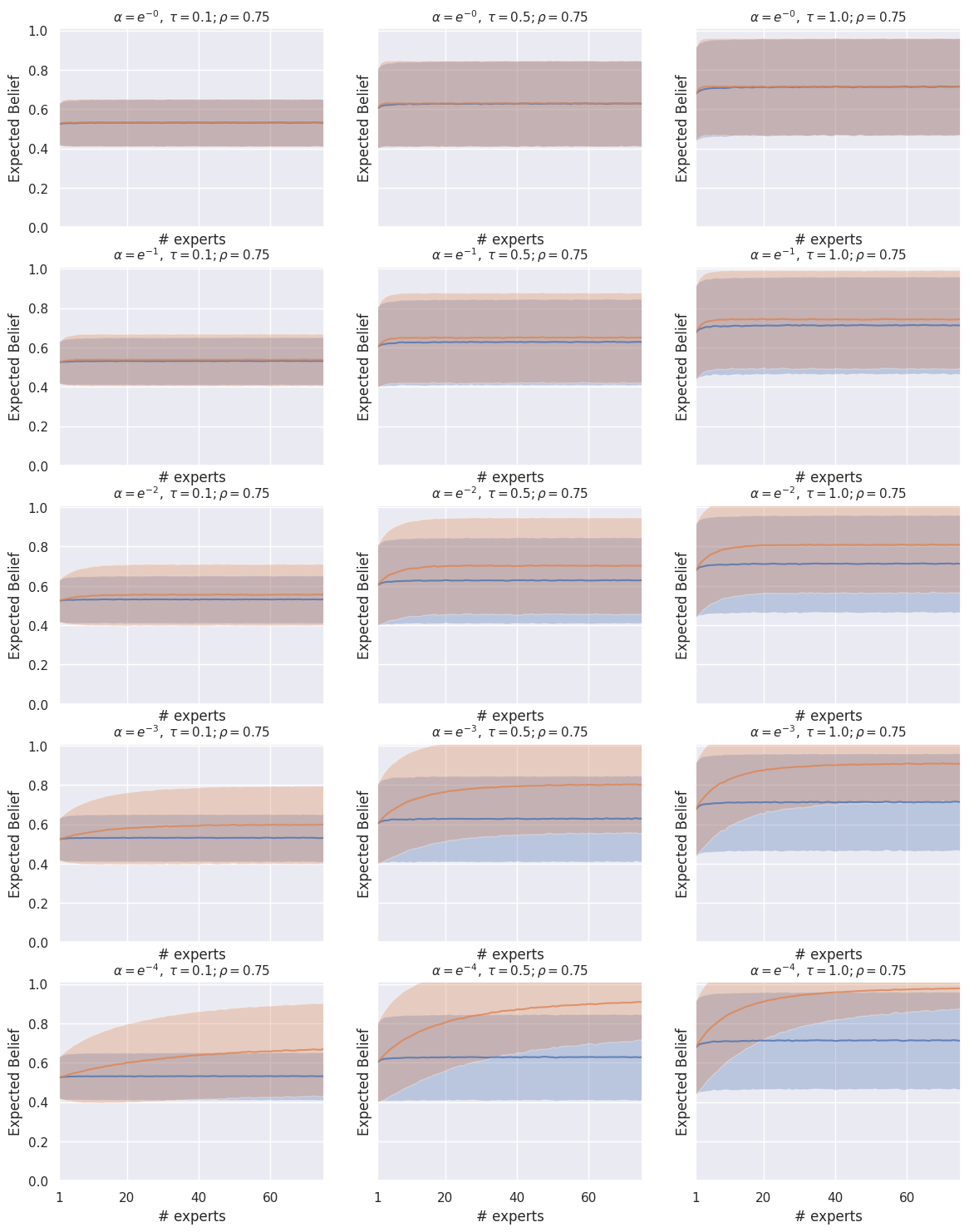}
    \caption{Expected belief (probability) as a function of the number of experts given rationales (orange) and without rationales (blue) when $\rho=0.75$ and $Y=1$, for various values of $\alpha$, $\tau$. Shaded area is 1 std deviation.}
\end{center}
\end{figure}

\begin{figure}[ht]\label{fig:expscore}
  \centering
  \includegraphics[width=0.9\textwidth]{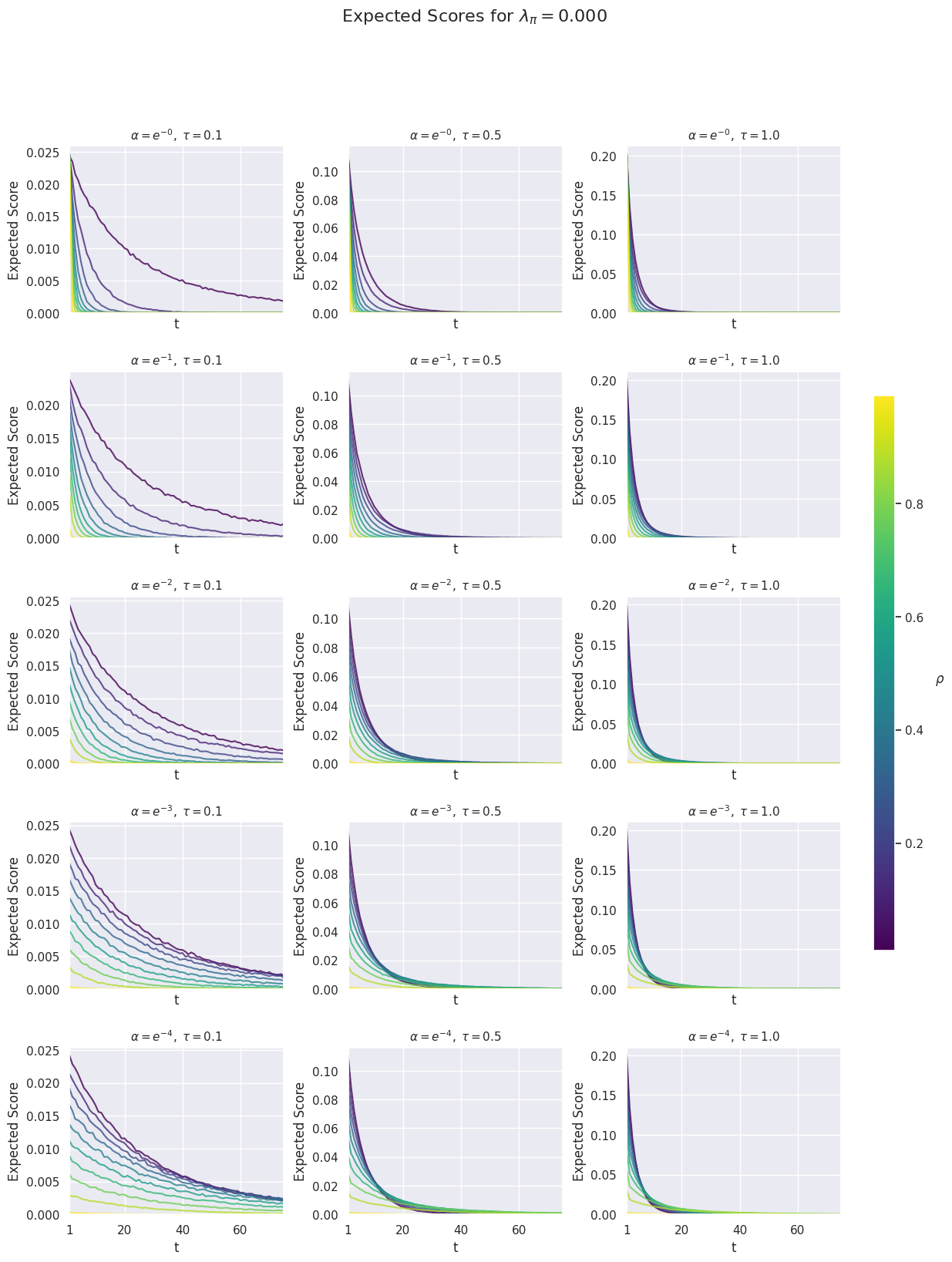}
  \caption{Numerical calculation of the ex-ante expected log score for expert $t$ when experts $1, \ldots, t-1$ report truthfully and submit their rationales, for different values of $\rho$, given $\alpha, \tau$. }
  \label{fig:your-label}
\end{figure}

\end{document}